\documentclass{article}
\usepackage{amsmath,amssymb}
\usepackage{bm}
\usepackage{verbatim}
\usepackage{wrapfig}
\usepackage{ascmac}
\usepackage{makeidx}
\usepackage{epsfig}
\usepackage{subfig}
\usepackage{graphicx}
\usepackage{amsthm}

\newcommand{\bbZ}{\mathbb{Z}}
\newcommand{\bbN}{\mathbb{N}}

\def\vc#1{\mbox{\boldmath $#1$}}
\newtheorem{thm}{Theorem}[section]

\newtheorem{dft}[thm]{Definition}
\newtheorem{lem}[thm]{Lemma}

\newtheorem{remark}{Remark}

\usepackage{doublespace}
\begin{document}
\title{Exact Solutions for M/M/c/Setup Queues}

\date{}
\author{
Tuan PHUNG-DUC \\
Department of Mathematical and Computing Sciences \\
Tokyo Institute of Technology \\
Email: tuan@is.titech.ac.jp
}
\maketitle
\begin{abstract}
Recently multiserver queues with setup times have been extensively studied because they have applications in power-saving data centers. A challenging model is the M/M/$c$/Setup queue where a server is turned off when it is idle and is turned on if there are some waiting jobs. Recently, Gandhi et al.~\cite{Gandhi13,Gandhi14} obtain the generating function for the number of jobs in the system using the recursive renewal reward approach. In this paper, we derive exact solutions for the joint stationary queue length distribution of the same model using two alternative methodologies: generating function approach and matrix analytic method. The generating function approach yields exact closed form expressions for the joint stationary queue length distribution and the conditional decomposition formula. On the other hand, the matrix analytic approach leads to an exact recursive algorithm to calculate the joint stationary distribution and performance measures so as to provide some application insights.
\end{abstract}

\section{Introduction}
The core part of cloud computing is data center where a large number of servers are available. These servers consume a large amount of energy. Thus, the key issue for the management of these server farms is to minimize the power consumption while keeping acceptable service level for users. It is reported that under the current technology an idle server still consumes about 60\% of its peak when processing jobs~\cite{Barroso07}. A natural suggestion to save power is to turn off idle servers. However, off servers need some setup time to be active during which they consume power but cannot process jobs. Thus, there exists a trade-off between power-saving and performance. This motivates the study of multiserver queues with setup times.  

Although queues with setup times have been extensively investigated in the literature, most papers deal with single server case~\cite{Takagi90,Bischof01,Choudhury98,Choudhury00}. These papers analyze single server queues with a general service time distribution. Artalejo et al.~\cite{Artalejo05} present an analysis for the multiserver queue with setup times where the authors consider the case in which at most one server can be in the setup mode at a time. This policy is later referred to as staggered setup in the literature~\cite{Gandhi10}. Artalejo et al.~\cite{Artalejo05} derive an analytical solution by solving the set of balance equations for the joint stationary distribution of the number of active servers and that of jobs in the system using a difference equation approach. The solution of the staggered setup model is significantly simplified by Gandhi et al.~\cite{Gandhi10}.

Recently,  motivated by applications in data centers, multiserver queues with setup times have been extensively investigated in the literature. In particular, Gandhi et al.~\cite{Gandhi10} extensively analyze multiserver queues with setup times. They obtain some closed form approximations for the ON-OFF policy where any number of servers can be in the setup mode at a time. As is pointed out in Gandhi et al.~\cite{Gandhi10}, from an analytical point of view the most challenging model is the ON-OFF policy where the number of servers in setup mode is not limited. Recently, Gandhi et al.~\cite{Gandhi13,Gandhi14} analyze the M/M/$c$/Setup model with the ON-OFF policy using a recursive renewal reward approach. Gandhi et al.~\cite{Gandhi13,Gandhi14} obtain the generating function of the number of jobs in the system and investigate the response time distribution.

The main aim  of our current paper is to derive explicit solutions for the joint queue length distribution for the M/M/$c$/Setup model with ON-OFF policy via two standard methodologies, i.e., generating function approach and matrix analytic method. The advantage of the generating function approach is that it provides detailed results for the joint stationary distribution, i.e., exact expressions for the joint stationary queue length distribution, generating functions and factorial moments of any order. Furthermore, the generating function approach gives a new look to the conditional decomposition for the queue length. On the other hand, the matrix analytic method yields an efficient algorithm where the rate matrix ($R$) and the first passage probability matrix ($G$) are explicitly obtained. In the two methods of this paper, we exploit special structure of the non-homogeneous part of the underlying Markov chain to have significant reductions of the computational complexity in comparison with existing methods in the literature~\cite{Gandhi13,Gandhi14,Benny_Johan11}.

Some closely related works are as follows. Mitrani~\cite{Mitrani11,Mitrani13} considers models for server farms with setup costs. The author analyzes the models where a group of reserve servers are shutdown instantaneously if the number of jobs in the system is smaller than some lower threshold and are powered up instantaneously when the queue length exceeds some upper threshold. Because of this instantaneous shutdown and setup, the underlying Markov chain in~\cite{Mitrani13} has a simple birth and death structure which allows closed form solutions. The author investigates the optimal lower and upper thresholds for the system. Mitrani~\cite{Mitrani11} extends~\cite{Mitrani13} to the case where each job has an exponentially distributed random timer exceeding which the job leaves the system. Schwartz et al.~\cite{Schwartz12} consider a similar model to that in~\cite{Mitrani11}. A finite buffer model is presented and analyzed in~\cite{phungduc15} while a model with impatient customers is analyzed in~\cite{phungduc14}.

The rest of this paper is organized as follows. Section~\ref{model:sec} presents the model in detail while Section~\ref{components:sec} is devoted to the analysis of the model via generating functions. Section~\ref{matrix_ana:sec} is devoted to the analysis via matrix analytic methods. Section~\ref{comparison_approach:sec} presents a comparison of the several approaches that can be used to analyze our M/M/c/Setup model. Section~6 presents some variant models for which the methodologies in this paper can be easily adapted. Some numerical examples are presented in Section~\ref{numerical:sec} to show insights into the performance of the system. Concluding remarks are presented in~Section~\ref{conclusion:sec}.

\section{Model and Markov Chain}\label{model:sec}
\subsection{Model}
We consider M/M/$c$/Setup queueing systems with ON-OFF policy. Jobs arrive at the system according to a Poisson process with rate $\lambda$. We assume that the service time of jobs follows an exponential distribution with mean $1/\mu$. In this system, upon service completion, a server is turned off immediately if there are no waiting jobs. Otherwise, it immediately takes a waiting job to process. Upon the arrival of a job, an OFF server (if any) is turned on and the job is placed in the buffer. However, a server needs some setup time to be active so as to serve waiting jobs. We assume that the setup time follows the exponential distribution with mean $1/\alpha$. Assuming that there are two jobs in the system, one job is receiving service and the other job in the buffer is waiting for a server in setup process. Under this situation, if the service completes before the setup, the waiting job is served immediately by the active server and the server in setup process is turned off. 

Let $j$ denote the number of customers in the system and $i$ denote the number of active servers. The number of servers in setup process is $\min(j-i, c-i)$. Under these assumptions, the number of active servers is smaller than or equal to the number of jobs in the system. Therefore, in this model a server is in either BUSY or OFF or SETUP. We assume that waiting jobs are served according to a first-come-first-served (FCFS) manner. We call this model an M/M/$c$/Setup queue. The exponential assumptions for the inter-arrival, setup time and service time allow us to construct a Markov chain whose stationary distribution is explicitly obtained. It should be noted that we can easily construct a Markov chain for a more general model with Markovian arrival process (MAP) and phase-type service and setup time distributions. However, the number of states of the resulting Markov chain explodes and thus analytical solutions do not exist. 

\subsection{Markov chain and notations}
It is easy to see that the stability condition for the system is $\lambda < c \mu$ because all the servers are eventually active if the number of 
jobs in the system is large enough. 
Let $C(t)$ and $N(t)$ denote the number of busy servers and the total number of jobs in the system, respectively. 
Under the assumptions made in Section~\ref{model:sec}, it is easy to see that $\{X(t) = (C(t), N(t)); t \geq 0\}$ forms a Markov chain in the state space 
\[
	\mathcal{S} = \{ (i,j); i = 0,1,\dots,c, j = i,i+1,\dots \}.
\]
%
%
See Figure~\ref{m:fig} for the transitions among states.
Let 
\[
\pi_{i,j} = \lim_{t \to \infty} \mathbb{P} (C(t) = i, N(t) = j), \qquad (i,j) \in \mathcal{S}.
\]
It should be noted that at the state $(i,j)$ the number of waiting jobs is $j-i$. 
We define the generating functions for the number of waiting jobs as follows.
\[
	\Pi_i (z) = \sum_{j=i}^\infty \pi_{i,j} z^{j-i}, \qquad i = 0,1,\dots,c.
\]
We are also interested in finding the factorial moments defined by $\Pi_i^{(n)} (1)$, where $f^{(n)} (x)$ denotes the $n$-th derivative of $f(x)$. We denote the set of non-negative integers and that of positive integers as follows.
\[
	\bbZ_+ = \{ 0,1,2,\dots \}, \qquad \bbN = \{ 1,2,3, \dots \}.
\]

\begin{dft}
For $\phi \in \mathbb{R}$, the Pochhammer symbol is defined as follows.
\[
	(\phi)_n = \left \{ 
	\begin{array}{ll}
	1 & n = 0, \\
	\phi (\phi + 1) \cdots (\phi + n-1), & n \in \bbN.
	\end{array}	
	\right.
\]
\end{dft}

\begin{figure}
\begin{center}
%
\includegraphics[scale=0.53]{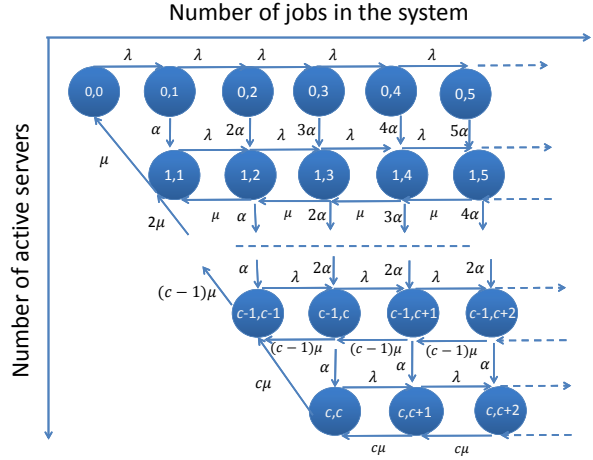}
\end{center}
\caption{State transition diagram.}
\label{m:fig}
%
\end{figure}

\section{Generating Function Approach}\label{components:sec}
In this section, we derive explicit expressions for the generating functions and the factorial moments. 
The term ``explicit" means that these expressions do not contain {\it limits} and they can be exactly calculated using a finite procedure.

\subsection{Explicit expressions}

The balance equations for the case $i=0$ read as follows. 
\begin{align}
 \label{pi_00:eq}
 \lambda \pi_{0,0} & = \mu \pi_{1,1},  \qquad j = 0, \\
 \label{pi_0jleqc1:eq}
(\lambda + j \alpha) \pi_{0,j} & =  \lambda \pi_{0,j-1} , \qquad j = 1,2,\dots,c-1, \\
 \label{pi_0jgeqc:eq} 
(\lambda + c \alpha) \pi_{0,j} & = \lambda \pi_{0,j-1}, \qquad  j \geq c.
\end{align}
Let $\widehat{\Pi}_0 (z) = \sum_{j=c}^\infty \pi_{0,j} z^j$.
Multiplying (\ref{pi_0jgeqc:eq}) by $z^j$ and summing over $j \geq c$, we obtain 
\begin{equation}\label{widehatp0z:eq}
  \widehat{\Pi}_0 (z) = \frac{  \lambda \pi_{ 0, c-1} z^c }{ \lambda + c \alpha - \lambda z} = z^c \frac{A_{0,0}}{\hat{z}_0 - z}, \qquad \Pi_0 (z) = \sum_{j=0}^{c-1} \pi_{0,j} z^j + \widehat{\Pi}_0 (z),
\end{equation}
where 
\[
	A_{0,0} = \lambda \pi_{0,c-1}, \qquad \hat{z}_0 = \frac{\lambda + c\alpha}{\lambda}.
\]
Equation (\ref{pi_0jleqc1:eq}) yields 
\[
	\pi_{0,j} = \pi_{0,0} \prod_{i=0}^j \frac{ \lambda }{ \lambda + j \alpha  }, \qquad j = 1,2,\dots,c-1.
\]
Furthermore, from the first equation in (\ref{widehatp0z:eq}), we obtain 
\[
	\pi_{0,j} =  \frac{  \lambda \pi_{ 0, c-1} }{ \lambda + c\mu} \left(  \frac{\lambda}{\lambda + c\alpha} \right)^{j-c} = \frac{A_{0,0}}{\hat{z}_0} \left(  \frac{1}{\hat{z}_0} \right)^{j-c}, \qquad j \geq c.
\]
\begin{remark}
At this moment, we have the fact that $\pi_{0,j} $ $(j \geq 1)$ and $\pi_{1,1}$ are expressed in terms of $\pi_{0,0}$.
\end{remark}
Differentiating (\ref{widehatp0z:eq}) $n$ times yields the following recursive formulae for the factorial moments. 
\begin{align*}
	\widehat{\Pi}_0^{(n)} (1) & = \frac{\lambda}{c\mu} \widehat{\Pi}_0^{(n-1)} (1) + \frac{\lambda}{c\mu} \pi_{0,c-1} (c-n)_n, \\
	{\Pi}_0^{(n)} (1)  & =  \sum_{j = 0}^{c-1} \pi_{0,j} (j-n+1)_n + \widehat{\Pi}_0^{(n)} (1),
\end{align*}
for $n \in \bbN$.

We shift to the case $i=1$. The balance equations are given as follows. 
\begin{align}
\label{pi1j_j1:eq}
(\lambda + \mu) \pi_{1,1} & = \alpha \pi_{0,1} + \mu \pi_{1,2} + 2\mu \pi_{2,2}, \\
\label{pi1j:eq_original}
(\lambda + \mu + (j-1) \alpha) \pi_{1,j} & = j \alpha \pi_{0,j} + \lambda \pi_{1,j-1} + \mu \pi_{1,j+1}, \qquad 2 \leq j \leq c-1, \\
\label{pi1j_jgeqc}
(\lambda + \mu + (c-1)\alpha ) \pi_{1,j} & = c \alpha \pi_{0,j} + \lambda \pi_{1,j-1} + \mu \pi_{1,j+1}, \qquad j \geq c.
\end{align}
Letting $\widehat{\Pi}_1 (z) = \sum_{j=c}^\infty \pi_{1,j} z^{j-1}$, we have $\Pi_1 (z) = \sum_{j=1}^{c-1} \pi_{1,j} z^{j-1} + \widehat{\Pi}_1 (z)$.
Multiplying (\ref{pi1j_jgeqc}) by $z^{j-1}$ and summing up over $j \geq c$ yields, 
\begin{equation}
(\lambda + \mu + (c-1) \alpha) \widehat{\Pi}_1(z) = \frac{c \alpha}{z} \widehat{\Pi}_0 (z) + \lambda z \widehat{\Pi}_1 (z) + \lambda \pi_{1,c-1} z^{c-1} + \frac{\mu}{z} (\widehat{\Pi}_1(z) - \pi_{1,c} z^{c-1}). 
\end{equation}
Rearranging this equation we obtain 
\begin{eqnarray}\label{pi1z:eq}
[ (\lambda + \mu + (c-1) \alpha) z - \lambda z^2 -\mu  ] \widehat{\Pi}_1 (z)  =   c \alpha \widehat{\Pi}_0 (z) + \lambda \pi_{1,c-1} z^c - \mu \pi_{1,c} z^{c-1}.
\end{eqnarray}
Let $f_1 (z) = (\lambda + \mu + (c-1) \alpha) z - \lambda z^2 -\mu$. 
Because $f_1(0) = -\mu < 0$, $f_1(1) = (c-1)\alpha >0$ and $f_1(\infty) = - \infty$, $f_1(z)$ has two roots $z_1$ and $\hat{z}_1$ such that $0 <  z_1 < 1 < \hat{z}_1$. We have 
\begin{align*}
	z_1 & = \frac{  \lambda + \mu + (c-1) \alpha - \sqrt{ (\lambda + \mu + (c-1) \alpha)^2 - 4 \lambda \mu }   }{2 \lambda}, \\
	\hat{z}_1 & = \frac{  \lambda + \mu + (c-1) \alpha + \sqrt{ (\lambda + \mu + (c-1) \alpha)^2 - 4 \lambda \mu }   }{2 \lambda}.
\end{align*}
Substituting $z=z_1$ into (\ref{pi1z:eq}), we obtain 
\begin{equation}
\label{pi1c:eq}
 \pi_{1,c} = \frac{ c \alpha \widehat{\Pi}_0 (z_1) + \lambda \pi_{1,c-1} z_1^{c}  }{ \mu z_1^{c-1} }.
\end{equation}
We derive a recursive scheme to determine $\pi_{1,j}$ ($j = 2,3,\dots,c$). Indeed, rewriting (\ref{pi1c:eq}) yields 
\[
	\pi_{1,c} = a^{(1)}_c + b^{(1)}_c \pi_{1,c-1},
\]
where 
\begin{equation}\label{a1cb1c:eq}
	a^{(1)}_c = \frac{c \alpha \widehat{\Pi}_0 (z_1)}{\mu z_1^{c-1}}, \qquad b^{(1)}_c = \frac{\lambda z_1}{\mu}.
\end{equation}
Using mathematical induction, we obtain the following lemma.

\begin{lem}\label{lemma1:lem}
\begin{equation}\label{pi1j:eq}
	\pi_{1,j} = a^{(1)}_j + b^{(1)}_j \pi_{1,j-1}, \qquad 2 \leq j \leq c, 
\end{equation}
where 
\begin{equation}\label{a1j:b1j:eq}
	a^{(1)}_j = \frac{ j \alpha \pi_{0,j}  }{ \lambda + \mu + (j-1) \alpha - \mu b^{(1)}_{j+1}  } , \qquad b^{(1)}_j = \frac{ \lambda }{ \lambda + \mu + (j-1) \alpha - \mu b^{(1)}_{j+1}  }, 
\end{equation}
for $j = c-1,c-2,\dots,1$. Furthermore, we have 
\[
	0 < a^{(1)}_j, \qquad 0 < b^{(1)}_j < \frac{\lambda}{\mu}, \qquad j = 1,2\dots,c.
\]
The generating function $\widehat{\Pi}_1 (z)$ is explicitly obtained as follows.
\begin{equation}\label{pi_1hat(z):eq}
\widehat{\Pi}_1 (z) = z^{c-1} \left( \frac{A_{1,0}}{\hat{z}_0 - z} + \frac{A_{1,1}}{\hat{z}_1 - z} \right),
\end{equation}
where 
\[
	A_{1,0} = \frac{A_{0,0} \hat{z}_0}{f_1(\hat{z}_0)}, \qquad A_{1,1} = -\frac{A_{0,0} \hat{z}_0}{f_1(\hat{z}_0)} + \pi_{1,c-1}.
\]
\end{lem}
\begin{proof}
We use mathematical induction for the proof of this lemma. 
First, we prove (\ref{pi1j:eq}). It is clear that (\ref{pi1j:eq}) is true for $j=c$ due to (\ref{a1cb1c:eq}). Assuming that (\ref{pi1j:eq}) is true for $j+1$, i.e., 
\[ 
\pi_{1,j+1} = a^{(1)}_{j+1} + b^{(1)}_{j+1}\pi_{1,j},
\]
 for some $j \leq c-1$. Substituting this expression into (\ref{pi1j:eq_original}) and rearranging the result we obtain (\ref{pi1j:eq}). 
Next, we also prove the inequalities. It is clear that Lemma~\ref{lemma1:lem} is true for $j=c$ since 
\[
	0 < a^{(1)}_c, \qquad 0 < b^{(1)}_c  < \frac{\lambda}{\mu}, 
\]
because $0 < z_1 < 1$. 
Assuming that $0 < b^{(1)}_{j+1}  < \frac{\lambda}{\mu}$ and $a^{(1)}_{j+1} > 0$, we have 
\[
	\mu + (c-1)\alpha < \lambda + \mu + (c-1) \alpha - \mu b^{(1)}_{j+1} < \lambda + \mu + (c-1) \alpha,
\]
which together with (\ref{a1j:b1j:eq}) yield
\[
	0 < \frac{\lambda}{\lambda + \mu + (j-1) \alpha} < b^{(1)}_j < \frac{\lambda}{\mu + (j-1) \alpha} < \frac{\lambda}{\mu}.
\]
and 
\[
	0 < \frac{j \alpha \pi_{0,j}}{\lambda + \mu + (j-1)\alpha} < a^{(1)}_j.
\]

Substituting (\ref{pi1c:eq}) to (\ref{pi1z:eq}) and dividing both sides by $(z-z_1)$, we obtain (\ref{pi_1hat(z):eq}) after some rearrangement. It should be noted that (\ref{decompose:formula}) 
is used to decompose $\Pi_1(z)$ into simple form.
\begin{equation}\label{decompose:formula}
\frac{1}{(a-z)(b-z)} = \frac{1}{b-a} \left( \frac{1}{a-z} - \frac{1}{b-z} \right), \qquad \forall \ a \neq b.
\end{equation}
\end{proof}

\begin{remark}
At this moment, $\pi_{1,j}$ ($j \geq 1$) is expressed in terms of $\pi_{0,0}$.
Thus, $\pi_{2,2}$ is also expressed in terms of $\pi_{0,0}$ due to the following formula representing the balance between the rates in and out the set $\{(i,j); i = 0,1; j \geq i \}$, i.e., 
\[
	2\mu \pi_{2,2} = \sum_{j=2}^\infty \min(j-1,c-1)\alpha \pi_{1,j}.
\]
\end{remark}
We are interested in finding the factorial moments. Taking the derivative of (\ref{pi1z:eq}) $n$ times yields
\begin{eqnarray}
\lefteqn{ f_1(z) \widehat{\Pi}_1^{(n)} (z) +  n f_1^\prime (z) \widehat{\Pi}_1^{(n-1)} (z) + \frac{n(n-1)}{2} f_1^{\prime \prime} (z) \widehat{\Pi}_1^{(n-2)} (z) = } \nonumber \\
&& c \alpha \widehat{\Pi}_0^{(n)} (z) + \lambda \pi_{1,c-1} (c-n+1)_n z^{c-n}  - \mu \pi_{1,c} (c-n)_n z^{c-1-n}.
\end{eqnarray}
Putting $z=1$ into this equation yields,
\begin{eqnarray}\label{widehatPi1:eq}
	\widehat{\Pi}_1^{(n)} (1) &  = &  \frac{c}{c-1} \widehat{\Pi}_0^{(n)} (1) + \frac{ n (\lambda -\mu - (c-1) \alpha) \widehat{\Pi}_1^{(n-1)} (1) + \lambda n(n-1) \widehat{\Pi}_1^{(n-2)} (1)   }{(c-1)\alpha} \nonumber \\
                                    &    & \mbox{} + \frac{\lambda \pi_{1,c-1} (c-n+1)_n  - \mu \pi_{1,c} (c-n)_n  }{(c-1)\alpha}, 
\end{eqnarray}
which is a recursive formula for computing $\widehat{\Pi}_1^{(n)} (1)$ ($n \in \bbN$). 
It should be noted that $\widehat{\Pi}_0^{(n)} (1)$ is explicitly obtained from (\ref{widehatp0z:eq}). Thus, from (\ref{widehatPi1:eq}) we obtain the factorial moments ${\Pi}_1^{(n)} (1)$.

Now, we consider general case where $i = 2,3,\dots,c-1$. The balance equations are as follows.
\begin{align}
\label{pi_ii_jleqc:eq}
(\lambda + i \mu) \pi_{i,i} & = \alpha \pi_{i-1,i} + i\mu \pi_{i,i+1} + (i+1)\mu \pi_{i+1,i+1}, \quad j = i \\ 
\label{pi_ij_jleqc:eq}
(\lambda + i \mu + (j-i) \alpha) \pi_{i,j} & = \lambda \pi_{i,j-1} + (j-i+1) \alpha \pi_{i-1,j} + i \mu \pi_{i,j+1},  \qquad i+1 \leq j \leq c-1, \\
\label{pi_ij_jgeqc:eq}
(\lambda + i\mu + (c-i) \alpha)  \pi_{i,j} & = \lambda \pi_{i,j-1} + (c-i+1) \alpha \pi_{i-1,j} + i\mu \pi_{i,j+1},  \qquad j \geq c.
\end{align}
We define the generating function $\widehat{\Pi}_i (z) = \sum_{j=c}^\infty \pi_{i,j-i} z^{j-i}$. We then have $\Pi_i (z) =  \sum_{j=i}^{c-1} \pi_{i,j} z^{j-i}  +  \widehat{\Pi}_i (z) $.
Multiplying (\ref{pi_ij_jgeqc:eq}) by $z^{j-i}$ and summing over $j \geq c$, we obtain 
\begin{eqnarray}
(\lambda + i \mu + (c-i) \alpha) \widehat{\Pi}_i (z) & = & \lambda \pi_{i,c-1} z^{c-i} + \lambda z \widehat{\Pi}_i (z) + \frac{(c-i+1) \alpha}{z} \widehat{\Pi}_{i-1} (z)  \nonumber \\
                                                                       &     & \mbox{} + \frac{i \mu }{ z} (\widehat{\Pi}_i (z)  - \pi_{i,c} z^{c+1-i}). 
\end{eqnarray}
Rearranging this equation, we obtain
\begin{eqnarray} \label{Pi(z):eq}
 [(\lambda + i \mu + (c-i) \alpha) z - \lambda z^2 - i \mu ]  \widehat{\Pi}_i (z)   =  (c-i+1) \alpha \widehat{\Pi}_{i-1} (z) + \lambda \pi_{i,c-1} z^{c-i+1} - i \mu \pi_{i,c} z^{c-i}.
\end{eqnarray}
Let $f_i (z)  = (\lambda + i \mu + (c-i) \alpha) z - \lambda z^2 - i \mu $. Because $f_i (0) = -i \mu < 0$, $f_i (1) = (c-i) \alpha > 0$ and ($f_i (\infty) = - \infty$), there exists some $0 < z_i < 1 < \hat{z}_1$ 
such that $f_i (z_i) = f_i (\hat{z}_i) = 0$. In particular, we have
\begin{align*}
	z_i & = \frac{  \lambda + i \mu + (c-i) \alpha - \sqrt{ (\lambda + i \mu + (c-i) \alpha)^2 - 4 i \lambda \mu }   }{2 \lambda}, \\
	\hat{z}_i & = \frac{  \lambda + i \mu + (c-i) \alpha + \sqrt{ (\lambda + i \mu + (c-i) \alpha)^2 - 4 i \lambda \mu }   }{2 \lambda}.
\end{align*}
Putting $z= z_i$ into (\ref{Pi(z):eq}) yields,
\begin{equation}\label{pi_ic:pi_ic_minus}
\pi_{i,c} = \frac{(c-i+1) \alpha \widehat{\Pi}_{i-1} (z_i) + \lambda \pi_{i,c-1} z_i^{c-i+1} }{i\mu z_i^{c-i} }
\end{equation}
This equation together with (\ref{pi_ij_jleqc:eq}) determine $\pi_{i,j}$ ($ i+1 \leq j \leq c$) as follows.
\begin{lem}\label{lemma32}
We have 
\[
\pi_{i,j} = a^{(i)}_j + b^{(i)}_j \pi_{i,j-1}, \qquad j = i+1, i+2,\dots,c,
\]
where 
\[
	a^{(i)}_c = \frac{ (c-i+1) \alpha \widehat{\Pi}_{i-1} (z_i)} { i\mu z_i^{c-i}  }, \qquad b^{(i)}_c = \frac{\lambda z_i }{i \mu},
\]
and for $j = c-1,\dots,i+1$, 
\[
	a^{(i)}_j = \frac{ (j-i+1)\alpha \pi_{i-1,j} +   i\mu a^{(i)}_{j+1}}{ \lambda + i\mu + (j-i)\alpha - i\mu b^{(i)}_{j+1} }, \qquad b^{(i)}_j = \frac{\lambda}{\lambda + i\mu + (j-i)\alpha - i\mu b^{(i)}_{j+1}}.
\]
Furthermore, we have
\[
	0 < a^{(i)}_j, \qquad 0 < b^{(i)}_j  < \frac{\lambda}{i\mu}.
\]
In addition, the generating function $\widehat{\Pi}_i (z)$ ($i=2,\dots,c-1$) is explicitly obtained as follows.
\begin{equation}\label{explicit:Pi(z)}
	\widehat{\Pi}_i (z) = z^{c-i} \left( \sum_{j=0}^{i} \frac{A_{i,j}} { \hat{z}_j - z} \right),
\end{equation}
where 
\[
	A_{i,j}= \frac{A_{i-1,j} \hat{z}_j}{f_i (\hat{z}_j) }, \qquad A_{i,i} = -(c-i+1) \alpha \sum_{j=0}^{i-1} \frac{ A_{i-1,j} \hat{z}_j   }{ f_i (\hat{z}_j) } + \pi_{i,c-1}.
\]
\end{lem}
\begin{proof}
The proof of Lemma~\ref{lemma32} proceeds in the same manner as used in Lemma~\ref{lemma1:lem}. 
We prove (\ref{explicit:Pi(z)}) using mathematical induction. Indeed, substituting 
\[
	\widehat{\Pi}_{i-1} (z) = z^{c-i+1} \left( \sum_{j=0}^{i-1} \frac{A_{i-1,j}} { \hat{z}_j - z} \right),
\]
into (\ref{Pi(z):eq}), deleting $(z-z_i)$ from both sides and rearranging the result, we obtain (\ref{explicit:Pi(z)}). It should be noted that (\ref{decompose:formula}) is used to obtain (\ref{explicit:Pi(z)}).
\end{proof}
\begin{remark}
It should be noted that $\pi_{i,j}$ ($j \geq i$) is expressed in terms of $\pi_{0,0}$. Furthermore, $\pi_{i+1,i+1}$ is expressed in terms of $\pi_{i,j}$ ($j = i+1,i+2,\dots$) and then in terms of $\pi_{0,0}$ via the balance of the flows in and out the set of states $\{(k,j); 0 \leq k \leq i, j \geq k \}$, i.e., 
\[
	(i+1)\mu \pi_{i+1,i+1} = \sum_{j=i+1}^\infty \min(j-i,c-i) \alpha \pi_{i,j}.
\]
\end{remark}
Taking the derivative of (\ref{Pi(z):eq}) $n$ times yields
\begin{eqnarray}
	\widehat{\Pi}_i^{(n)} (1) &  = &  \frac{c -i +1}{c-i} \widehat{\Pi}_{i-1}^{(n)} (1)  + \frac{ n (\lambda -\mu - (c-i) \alpha) \widehat{\Pi}_i^{(n-1)} (1) + n(n-1) \lambda  \widehat{\Pi}_i^{(n-2)} (1)   }{(c-i)\alpha} \nonumber \\
                                    &    & \mbox{} + \frac{\lambda \pi_{i,c-1} (c -i +2 -n)_n   - i \mu \pi_{i,c} (c-i + 1- n )_n   }{(c-i)\alpha},
\end{eqnarray}
which is a recursive formula to compute all the factorial moments $\widehat{\Pi}_i^{(n)} (1)$ ($n \in \bbN$). It should be noted that $\widehat{\Pi}_i^{(0)} (1) = \widehat{\Pi}_i (1)$ and $\widehat{\Pi}_{i-1}^{(n)} (1)$ ($n \in \bbN$) are already known.

Finally, the case $i = c $ needs some special treatment. 
Balance equations read as follows. 
\begin{eqnarray}
(\lambda + c\mu ) \pi_{c,c} & = & \alpha \pi_{c-1,c} + c\mu \pi_{c,c+1}, \qquad j = c, \\
\label{picj:eq}
(\lambda + c\mu ) \pi_{c,j}  & = & \alpha \pi_{c-1,j} + \lambda \pi_{c,j-1} + c\mu \pi_{c,j+1}, \qquad j \geq c + 1.
\end{eqnarray}
Defining  
\[
	\widehat{\Pi}_c ( z) = \sum_{j = c}^\infty \pi_{c,j} z^{j-c}, 
\]
we have $\Pi_c (z) = \widehat{\Pi}_c ( z)$.
Multiplying (\ref{picj:eq}) by $z^{j-c}$ and summing up over $j \geq c$ yields
\begin{equation}\label{Pc(z):eq}
	(\lambda + c\mu) \widehat{\Pi}_c(z) = \frac{\alpha}{z} \widehat{\Pi}_{c-1} (z) + \lambda z \widehat{\Pi}_c(z) + \frac{c\mu}{z} (\widehat{\Pi}_c(z) - \pi_{c,c}),
\end{equation}
leading to 
\[
	f_c(z) \widehat{\Pi}_c(z) = \alpha \widehat{\Pi}_{c-1}(z) - c\mu \pi_{c,c},
\]
or equivalently,
\begin{align}\label{pi_cz:eq}
	\widehat{\Pi}_c(z) & = \frac{ \alpha \widehat{\Pi}_{c-1} (z) - c \mu \pi_{c,c}  }{ z- 1} \frac{1}{c\mu - \lambda z} =  \frac{ \alpha \left( \widehat{\Pi}_{c-1} (z) - \widehat{\Pi}_{c-1} (1) \right)  }{ z- 1} \frac{1}{c\mu - \lambda z}.
\end{align}
where $	f_c(z) = (\lambda + c\mu) z - \lambda z^2 - c\mu$ and $ \alpha \widehat{\Pi}_{c-1} (1) = c \mu \pi_{c,c}$ is used in the second equality of (\ref{pi_cz:eq}).

It should be noted that the numerator and denominator of the first term in the right hand side of (\ref{pi_cz:eq}) vanish at $z=1$. Thus, applying l'Hopital's rule, we obtain 
\[
	\widehat{\Pi}_c(1) = \frac{ \alpha \widehat{\Pi}_{c-1}^\prime (1) }{ c\mu - \lambda}.
\]
Substituting $\widehat{\Pi}_{c-1}(z)$ in the form of (\ref{explicit:Pi(z)}) with $i=c-1$ into (\ref{pi_cz:eq}), we obtain 
\begin{equation}\label{pi_c(z):decompose}
	\widehat{\Pi}_c(z) = \sum_{j=0}^c \frac{A_{c,j}}{\hat{z}_j - z},
\end{equation}
where 
\[
	\hat{z}_c = \frac{c\mu}{\lambda}, \qquad A_{c,j} = \frac{A_{c-1,j}}{\hat{z}_c - 1}, \qquad j = 0,1,\dots,c-1, \qquad A_{c,c}= - \sum_{j=0}^{c-1} \frac{A_{c-1,j} \hat{z}_j }{f_c(\hat{z}_j)}.
\]

Taking the derivative of (\ref{Pc(z):eq}) $n$ times and rearranging the result and then applying l'Hopital's rule yields,

\[
	\Pi^{(n)}_c (1) = \frac{ \alpha \Pi_{c-1}^{(n+1)} (1) + \lambda n (n-1) \Pi_c^{(n-2)} (1) + 2\lambda n \Pi_c^{(n-1)} (1)   }{ (n+1) (c\mu - \lambda)}.   
\]   
It should be noted that $\Pi_{c-1}^{(n+1)} (1)$ and $\Pi_c^{(0)} (1) = \Pi_c (1)$ are already given.

At this moment, all the probabilities $\pi_{i,j}$ ($j \leq c$)  and the generating functions $\widehat{\Pi}_i (z)$ ($i=0,1,\dots,c$) are expressed in terms of $\pi_{0,0}$ which is uniquely determined using the 
following normalization condition. 
\[
	\Pi_0 (1) +  \Pi_1 (1) + \cdots + \Pi_c (1) = 1.
\]

\begin{remark}
Since explicit expressions for the generating functions are available, we can easily obtain explicit results for the factorial moments and the joint stationary distribution using $A_{i,j}$ ($0 \leq i\leq j \leq c$) and $\hat{z}_i$ ($i=0,1,\dots,c$). In particular, it follows from (\ref{explicit:Pi(z)}) and (\ref{pi_c(z):decompose}) that $\pi_{i,j}$ ($i=1,2,\dots,c$, $j \geq c$) is a linear combination of $1/\hat{z}_k^j$ ($k=0,1,\dots,i$).
\end{remark}

\begin{remark}
It was shown in~\cite{Levy_Yechiali:76} that $\hat{z}_i$ ($i=1,2,\dots,c-1$) are distinct. In the above analysis we implicitly assume that $\hat{z}_0 \neq \hat{z}_i$ ($i=1,2,\dots,c$) and $\hat{z}_c \neq \hat{z}_i$ ($i=0,1\dots,c-1$). 
In case where there exists some $i$ such that $\hat{z}_0 = \hat{z}_i$ ($i=1,2, \dots,c-1$) or (and) some $j$ such that $\hat{z}_j = \hat{z}_c$, we still have explicit expressions for the generating functions and the the joint stationary distribution after some minor modification. In particular, if $\hat{z}_0 = \hat{z}_i$ for some $i=1,2,\dots,c-1$, $\pi_{i,j}$ is a linear combination of $1/\hat{z}_k^j$ ($k=0,1,\dots,i-1$) and $j/\hat{z}_0^j$.
\end{remark}


\begin{remark}
The computational complexity of the generating function approach is $O(c^2)$. Indeed, we need to calculate $A_{i,j}$ and $\pi_{i,j}$ ($i\leq j, 0 \leq j \leq c$) in the following order:
\[
	(0,0) \rightarrow (0,1) \rightarrow \cdots \rightarrow (0,c) \rightarrow (1,1) \rightarrow (1,2) \rightarrow \cdots \rightarrow (1,c) \rightarrow \dots \rightarrow (c,c).
\] 
As a result, the complexity is of order $\sum_{i=0}^c i = c(c+1)/2 = O(c^2)$. 
It should be noted that the recursive procedure for $\pi_{i,j}$ ($0 \leq i \leq j \leq c$) is numerically stable since it involves only positive numbers, i.e., $a^{(i)}_j$ and $b^{(i)}_j$.
\end{remark}

\subsection{Conditional stochastic decomposition}\label{Conditional_Decomposition:sec}
We have derived the following result.
\begin{eqnarray*}
 \Pi_c(z) & = & \frac{ \alpha ( \Pi_{c-1} (z) - \pi_{c-1,c-1})    -c\mu \pi_{c,c}     }{ (z-1) (c\mu -\lambda z)  }, \\ 
 \Pi_c(1) & = & \frac{ \alpha \Pi_{c-1}^\prime (1)   }{c\mu - \lambda}.
\end{eqnarray*}
Let $Q^{(c)}$ denote the conditional queue length given that all $c$ servers are busy in the steady state, i.e.,
\[
	\mathbb{P} (Q^{(c)} = i) = \mathbb{P} (N(t) = i + c \ | \ C(t) = c).
\]
Let $P_c(z)$ denote the generating function of  $Q^{(c)}$. It is easy to see that 
\begin{align*}
	P_c(z) & = \frac{\Pi_c(z)}{\Pi_c(1)} \\
            & = \frac{  \alpha ( \Pi_{c-1} (z) - \pi_{c-1,c-1})    -c\mu \pi_{c,c}     }{\alpha \Pi_{c-1}^\prime(1) (z-1)} \frac{1- \rho}{1 -\rho z}\\
	        & = \frac{\Pi_{c-1} (z) - \Pi_{c-1}(1)}{\Pi_{c-1}^\prime(1)(z-1)}   \frac{1- \rho}{1 -\rho z} \\
	        & = \frac{\sum_{j=1}^\infty \pi_{c-1,c-1+j} (z^j - 1)  }{\Pi_{c-1}^\prime(1)(z-1)} \frac{1- \rho}{1 -\rho z} \\
	        & = \frac{\sum_{j=1}^\infty \pi_{c-1,c-1+j} \sum_{i=0}^{j-1} z^i  }{\Pi_{c-1}^\prime(1)} \frac{1- \rho}{1 -\rho z} \\
	        & = \frac{ \sum_{i=0}^\infty \left(\sum_{j = i+1}^\infty \pi_{c-1,c-1+j} \right)  z^i }{\Pi_{c-1}^\prime(1)} \frac{1- \rho}{1 -\rho z},
\end{align*}
where we have used $c\mu \pi_{c,c} = \alpha (\Pi_{c-1} (1) - \pi_{c-1,c-1}) $ in the second equality.

It should be noted that $(1-\rho)/(1-\rho z)$ is the generating function of the number of waiting jobs 
in the conventional M/M/$c$ system without setup times under the condition that $c$ servers are busy. 
We denote this random variable by $Q^{(c)}_{ON-IDLE}$. It should be noted that $Q^{(c)}_{ON-IDLE}$ can also be interpreted as the number of jobs in the M/M/1 queue without vacation where the arrival rate and the service rate are $\lambda$ and $c\mu$, respectively.   
We give a clear interpretation for the generating function
\[
	\frac{ \sum_{i=0}^\infty \left(\sum_{j = i+1}^\infty \pi_{c-1,c-1+j} \right)  z^i }{\Pi_{c-1}^\prime(1)}.
\]
For simplicity, we define  
\begin{eqnarray*}
	p_{c-1,i} =   \frac{\sum_{j = i+1}^\infty \pi_{c-1,c-1+j}}{\Pi_{c-1}^\prime(1)}, \qquad i \in \bbZ_+.
\end{eqnarray*}
We have
\[
	\sum_{j=i+1}^\infty \pi_{c-1,c-1 + j} = \mathbb{P} ( N(t) - C(t) > i \ | \ C(t) =  c-1) \mathbb{P} (C(t) = c-1),
\]
and 
\[
	\Pi_{c-1}^\prime(1) = \mathbb{E} [N(t) - C(t) \ | \ C(t) = c-1]  \mathbb{P} (C(t) = c-1).
\]
Thus, we have 
\[
	p_{c-1,i} = \frac{\mathbb{P} ( N(t) - C(t) > i \ | \ C(t) = c-1)}{ \mathbb{E} [N(t) - C(t) \ | \ C(t) = c-1] }.
\]

It should be noted that $N(t)-C(t)$ is the number of jobs in the system that are waiting for the last server (in setup mode) to be active. Thus, $p_{c-1,i}$ ($i=0,1,2,\dots$) represents distribution of the number of waiting customers in front of an arbitrary waiting customer (not being served) under the condition that $c-1$ servers are active and the last server is in setup mode (see Burke~\cite{Burke}). Let $Q_{Res}$ denote the random variable with the distribution $p_{c-1,i}$ ($i=0,1,2,\dots$). Our decomposition result is summarized as follows. 
\begin{equation}\label{conditional_decomposition}
	Q^{(c)} \,{\buildrel d \over =}\  Q^{(c)}_{ON-IDLE}  +  Q_{Res}. 
\end{equation}
We observe that $Q_{Res}$ represents the number of extra jobs due to the setup time.

\begin{remark}
The conditional decomposition (\ref{conditional_decomposition}) is not explicit in the sense that $Q^{(c)}_{ON-IDLE}$ is not an explicit random variable. However, it is useful for understanding the behavior of the system. This situation is the same in the decomposition of M/M/1 queue with working vacation (M/M/1/WV) by Servi and Finn~\cite{Servi02}. The reason for the ``implicit" stochastic decomposition is that the service is continued during the working vacation.
\end{remark}

\begin{remark}
Tian et al.~\cite{Tian99,Tian03b,Tian_Zhang} obtain a similar result for a multiserver model with vacation. However, the random variable with the distribution $p_{c-1,i}$ here is not given a clear physical meaning in~\cite{Tian99,Tian03b,Tian_Zhang}.
\end{remark}

\section{Matrix Analytic Methods}\label{matrix_ana:sec}
In this section we present an analysis of the model based on a quasi-birth-and-dearth process (QBD) approach. 
\subsection{QBD formulation}
The infinitesimal of $\{X(t) \}$ is given by 
\[
Q  = \left(
\begin{array}{ccccc}
     Q^{(0)}_0 		&    Q^{(0)}_{1}    &       O & O      &      \cdots     \\
     Q^{(1)}_{-1} 	&    Q^{(1)}_0       &    Q^{(1)}_1  & O      &      \cdots  \\
	 O      &    Q^{(2)}_{-1}     &   Q^{(2)}_0    & Q^{(2)}_1    &       \cdots    \\
	 O      &    O       &    Q^{(3)}_{-1}  & Q^{(3)}_0      &     \cdots    \\
     \vdots &    \vdots  &   \vdots & \vdots &      \ddots    \\
\end{array}
\right),
\]
where $O$ denotes the zero matrix with an appropriate dimension. A Markov chain with this type of block tridiagonal matrix is called a level dependent quasi-birth-and-death process for which 
some efficient algorithms are available~\cite{bright95,matrix-conf}. 
 The block matrices $Q^{(i)}_{-1}$ ($i \geq c+1$), $Q^{(i)}_0$ ($i \geq c$) and $Q^{(i)}_1$ ($i \geq c$) are independent of $i$ and are explicitly given as follows.
\[
	Q^{(i)}_{-1} = Q_{-1}  = diag (0, \mu, \dots, c \mu), \qquad  Q^{(i)}_1 = Q_{1}  = \lambda I.
\]
\begin{align*}
Q^{(i)}_0 = Q_0 & = \left(
\begin{array}{cccccc}
-q_0 &         c \alpha  &               0 & \cdots &            \cdots &       0 \\
     0      &       -q_1 &          (c-1) \alpha & \ddots &                   &  \vdots \\
          0 &                0 &          -q_2 & \ddots &            \ddots &  \vdots \\
          \vdots &          \ddots &          \ddots & \ddots &            \ddots & 0 \\
          \vdots &                 &          \ddots & \ddots & -q_{c-1} &   \alpha \\
               0 &          \cdots &          \cdots &      0 &           0 & -q_{c}
\end{array}
\right),
\end{align*}
where $q_j = \lambda + (c-j) \alpha + j \mu$.
Furthermore, $Q^{(i)}_2$ ($i \leq c$), $Q^{(i)}_1$ ($i \leq c-1$) and $Q^{(i)}_0$ ($i \leq c$) are $(i+1) \times i$, $(i+1) \times (i+1)$ and $(i+1) \times (i+2)$ matrices whose contents are given as follows. 
\begin{align*}
	Q^{(i)}_1  & =  \left(  
	\begin{array}{ccccc}
		\lambda &    0  &   \cdots  & 0 &  0 \\
		0      &       \lambda  &    \ddots & \vdots &    \vdots \\
        \vdots &     \ddots   &     \ddots & 0        &  0 \\
        0 &          \cdots &         0 & \lambda &  0 
	\end{array}
	\right), \qquad 
		Q^{(i)}_{-1}  = \left(
	\begin{array}{cccccc}
		0 &         0  &        \cdots  & \cdots &            0  \\
		0      &       \mu &          \ddots & \ddots &     \vdots            \\
        0 &                0 &           & \ddots &            \vdots  \\
        \vdots &          \ddots &          \ddots & \ddots &            0  \\
        \vdots &                 &          \ddots & \ddots &  (i-1) \mu   \\
        0        &          \cdots &          \cdots &      0 &  i \mu 
	\end{array}
	\right),	\\
	Q^{(i)}_0 & = \left(
	\begin{array}{cccccc}
		-q^{(i)}_0 &         i \alpha  &               0 & \cdots &            \cdots &       0 \\
		0      &       -q^{(i)}_1 &          (i-1) \alpha & \ddots &                   &  \vdots \\
        0 &                0 &          -q^{(i)}_2 & \ddots &            \ddots &  \vdots \\
        \vdots &          \ddots &          \ddots & \ddots &            \ddots & 0 \\
        \vdots &                 &          \ddots & \ddots & -q^{(i)}_{i-1} &   \alpha \\
        0        &          \cdots &          \cdots &      0 &           0 & -q^{(i)}_{i}
	\end{array}
	\right),		
\end{align*}
where $q^{(i)}_j = (i-j) \alpha + j \mu$ ($j=0,1,\dots,i$).
%
%
Let 
\begin{align*}
\vc{\pi}_i & = (\pi_{0,i}, \pi_{1,i}, \dots, \pi_{\min(i,c),i}), \qquad i \in \bbZ_+, \qquad \vc{\pi}  = (\vc{\pi}_0,\vc{\pi}_1,\dots).
\end{align*}
The stationary distribution $\vc{\pi}$ is the unique solution of 
\[
	\vc{\pi} Q = \vc{0}, \qquad \vc{\pi} \vc{e} = 1, 
\]
where $\vc{0}$ and $\vc{e}$ represent a row vector of zeros and a column vector of ones with an appropriate size. According to the matrix analytic method~\cite{Neuts81,Rdef}, we have
\[
	\vc{\pi}_i = \vc{\pi}_{i-1} R^{(i)}, \qquad i \in \bbN, 
\]
and $\vc{\pi}_0$ is the solution of the boundary equation 
\[
	\vc{\pi}_0 ( Q^{(0)}_0 + R^{(1)} Q^{(1)}_{-1} ) = \vc{0}, \qquad \vc{\pi}_0 ( I + R^{(1)} + R^{(1)} R^{(2)} + \cdots ) \vc{e} = 1. 
\]
Here $\{ R^{(i)}; i \in \bbN \}$ is the minimal nonnegative solution of the following equation 
\begin{equation}\label{ri_ri_plus:eq}
	Q^{(i-1)}_1 + R^{(i)} Q^{(i)}_0 +  R^{(i)} R^{(i+1)} Q^{(i+1)}_{-1} = O.
\end{equation}  

\subsection{Homogeneous part}
\subsubsection{The rate matrix}
It should be noted that $Q^{(i-1)}_1 = Q_1$ ($i \geq c$), $Q^{(i)}_0 = Q_0$ ($i \geq c$) and $Q^{(i)}_{-1} = Q_{-1}$ ($i \geq c+1$).
Thus, we have $R^{(i)} = R$ for $i \geq c+1$ and $R$ is the minimal nonnegative solution of the following equation.
\begin{equation}\label{homogeneous:Rmatrix}
	Q_{1} + R Q_{0} +  R^2 Q_{-1} = O.
\end{equation}  
We know that $R$ is an upper diagonal matrix, i.e., $R(i,j) = r_{i,j}$ ($j \geq i$) and $R(i,j) = 0$ if $j < i$ because $Q_{-1}, Q_0, Q_1$ are upper diagonal matrix. A similar structure is also found in the model in~\cite{Neuts81,Perel_Yechiali}. Furthermore, this type of QBD is considered in more general contexts in~\cite{Benny_Johan11}. Comparing the diagonal part of the quadratic equation above, we obtain 
\begin{equation}
\lambda -  (\lambda + i\mu + (c-i)\alpha ) r_{i,i}  + i \mu r_{i,i}^2 = 0, \qquad i = 0,1,\dots,c-1,c.
\end{equation}
which has two roots. Because $R$ is the minimal nonnegative solution of (\ref{homogeneous:Rmatrix}), we must choose the smallest root for $r_{i,i}$. Thus, we have 
\begin{equation}
r_{i,i} = \frac{ \lambda + i \mu + (c-i) \alpha - \sqrt{ (\lambda + i\mu + (c-i) \alpha)^2 - 4 i \lambda \mu   }               }{ 2 i \mu}, \qquad i = 1,2,\dots,c-1,
\end{equation}
and 
\[
	r_{0,0} = \frac{\lambda}{\lambda + c \alpha}, \qquad r_{c,c} = \frac{\lambda}{c \mu} < 1.
\]
Next, we shift to the non-diagonal elements, i.e., $r_{i,j}$ ($j > i$).
Comparing the $(i,j)$ element in the quadratic equation, we obtain 
\[
  (c-j+1) \alpha r_{i,j-1} - (\lambda + (c-j) \alpha + j \mu) r_{i,j} + j \mu \sum_{k=i}^j r_{i,k} r_{k,j} = 0.
\]
For $j = i+1$, we obtain 
\[
(c-i) \alpha r_{i,i} - (\lambda + (c-i-1) \alpha + (i+1) \mu) r_{i,i+1} + (i+1) \mu (r_{i,i} r_{i,i+1} + r_{i,i+1} r_{i+1,i+1}  ) = 0.
\]
Thus, 
\begin{align*}
	r_{i,i+1}  = & \frac{ (c-i) \alpha r_{i,i}   }{  \lambda + (c-i-1) \alpha + (i+1) \mu - (i+1) \mu (r_{i,i} + r_{i+1,i+1})  },  \\
               & i = 0,1,\dots,c-1.
\end{align*}
%
%
%
It should be noted that the right hand side contains only known quantities obtained in previous steps. 
For the general case, we have 
\[
	r_{i,j} = \frac{(c-j+1)\alpha r_{i,j-1} + j\mu \sum_{k=i+1}^{j-1} r_{i,k} r_{k,j}   }   {  \lambda + (c-j) \alpha + j\mu  - j \mu (r_{i,i} + r_{j,j})  }, \qquad j > i.
\]
We can rewrite this formula as follows.
\begin{align*}
	r_{i,i+h+1} = & \frac{(c-i-h)\alpha r_{i,i+h} + (i+h+1) \mu \sum_{k=i+1}^{i+h} r_{i,k} r_{k,i+h+1}   }   {  \lambda + (c-i-h-1) \alpha + (i+h+1)\mu  - (i+h+1) \mu (r_{i,i} + r_{i+h+1,i+h+1})  }, \\
	                 & i = 0,1,\dots,c-h-1, \qquad h = 0,1,\dots,c-1.
\end{align*}
From these recursive formulae, we can calculate the elements of the rate matrix from the diagonal part and then the upper diagonal parts consequently. 

\subsubsection{Non-homogeneous part}\label{Non-homogeneous:part}
Because $R^{(i)} = R$ ($i = c+1,c+2,\dots$) which has been explicitly obtained, we only need to find $R^{(i)}$ ($i=c,c-1,\dots,1$).
Indeed, $R^{(i)}$ ($i=c,c-1,\dots,1$) is easily obtained using the following backward formula. 
\[
	R^{(i)}  = - Q^{(i-1)}_1 \left( Q^{(i)}_0  +  R^{(i+1)} Q^{(i+1)}_{-1} \right)^{-1},  \qquad i = c,c-1,\dots,1.
\]
This is equivalent to solving the following system of linear equations. 
\[
	R^{(i)} \left( Q^{(i)}_0  +  R^{(i+1)} Q^{(i+1)}_{-1} \right) = - Q^{(i-1)}_1 ,  \qquad i = c,c-1,\dots,1.
\]

Due to the special structure of the rate matrices, i.e., they are upper diagonal matrices, this system of linear equations can be efficiently solved as follows. In this case, we need to solve the following equation  
\begin{equation}\label{backward:eq}
	X A = - Q^{(i-1)}_0, 
\end{equation}
where $A = Q^{(i)}_0  +  R^{(i+1)} Q^{(i+1)}_{-1} $ is an upper diagonal matrix of size $(i+1) \times (i+1)$ and and $X$ is also an upper diagonal matrix of size $i \times (i+1)$ matrix. 
Let $\vc{x}_j = (0,0,\dots,x_{j,j}, x_{j,j+1},\dots,x_{j,i})$ ($j=0,1,\dots,i-1$) denote the $j$-th row vector of $X$. The above equation is equivalent to 
\[
	\vc{x}_j A = (0,0,\dots, - \lambda, 0, \dots, 0), \qquad j = 0,1,\dots,i-1,
\]
where the $- \lambda$ is the $(j+1)$-th entry of the vector in the right hand side. The solution of this equation is given by 
\begin{align}\label{recursive:eq_xij}
	x_{j,j} & = - \frac{\lambda}{a_{j,j}}, \qquad
	x_{j,l}  = - \frac{\sum_{k=j}^{l-1} x_{j,k} a_{k,l}}{a_{l,l}}, \qquad l = j+1,j+2,\dots,i,
\end{align}
where $a_{i,j}$ is the $(i,j)$ entry of $A$.

\begin{remark}
The computational complexity of (\ref{recursive:eq_xij}) is $i-j$ and thus, the computational complexity for (\ref{backward:eq}) is $O(i^2) = \sum_{j=0}^{i-1} (i-j)$ instead of $O(i^3)$ by a direct inversion of $X$. Therefore, the computational complexity for obtaining the rate matrices $R^{(i)}$ ($i = 1,2,\dots,c$) in the non-homogeneous part is of the order of $O(c^3)$ because $\sum_{j=1}^c i^2 = O(c^3)$. It should be noted that if we solve (\ref{backward:eq}) by a direct inversion of the $X$, the computational complexity for $R^{(i)}$ is $i^3$ and thus the computational complexity for all the rate matrices in the non-homogeneous part ($R^{(i)}$, $i = 1,2,\dots,c$) is $O(c^4)$.  
\end{remark}

\subsection{The $G$-matrix}\label{homo_part_gmatrix}
In this section, we derive explicit expressions for the $G$-matrix of our QBD process. It should be noted that $G$-matrix records the first passage probabilities to one level left in the homogeneous part (i.e., the number of jobs in the system is greater than $c$). These probabilities are also obtained using the recursive renewal reward approach by~\cite{Gandhi13,Gandhi14}.  The $G$-matrix is the minimal and nonnegative solution of the following equation~\cite{Neuts81}.
\begin{equation}\label{equationforG:eq}
Q_{-1} + Q_0 G + Q_1 G^2 = O.
\end{equation}
From the physical interpretation of $G$, we see that $G$ is also an upper diagonal matrix.  
Using a similar method as in the case of $R$-matrix, we are able to obtain explicit expressions for all the elements of $G$. 
Let $g_{i,j}$ ($i,j=0,1,\dots,c$) denote the $(i,j)$ element of $G$. Comparing the element $(0,0)$ in both sides of (\ref{equationforG:eq}) yields,
\[
	-(\lambda + c\alpha) g_{0,0} + \lambda g_{0,0}^2 = 0. 
\]
Since $0 \leq g_{0,0} \leq 1$, we obtain $g_{0,0} = 0$. Equating the $(i,i)$ ($i = 1,2,\dots,c-1$) elements in both sides of (\ref{equationforG:eq}), we obtain 
\[
	i\mu - (\lambda + (c-i) \alpha + i\mu) g_{i,i} + \lambda g_{i,i}^2 = 0, \qquad i = 1,2,\dots,c-1.
\]
Combining with the condition that $0 \leq g_{i,i} \leq 1$, we obtain 
\[
	g_{i,i} =  \frac{  \lambda + i \mu + (c-i) \alpha - \sqrt{ (\lambda + i \mu + (c-i) \alpha)^2 - 4 i \lambda \mu }   }{2 \lambda},
\]
which is identical to $z_i$. Finally, comparing the $(c,c)$ elements in both sides of (\ref{equationforG:eq}) we obtain 
\[
	c\mu -(\lambda + c\mu )g_{c,c} + \lambda g_{c,c}^2 = 0,
\]
which has two roots $1$ and $\lambda/(c\mu)$. Because $g_{c,c}$ is the minimal solution of this equation, we have $g_{c,c} = \lambda/(c\mu)$.
We have obtained all the diagonal elements of the $G$-matrix. Using the same manner as for $R$-matrix, we also recursively obtain the upper diagonal elements.
First, we obtain the upper diagonal elements $g_{i,i+1}$ ($i=0,1,\dots,c-1$). Indeed, comparing the elements $(i,i+1)$ in both sides of (\ref{equationforG:eq}), we obtain 
\begin{equation}
-q_i g_{i,i+1} + (c-i)\alpha g_{i+1,i+1} + \lambda (g_{i,i} g_{i,i+1} + g_{i,i+1} g_{i+1,i+1}) = 0,
\end{equation}
leading to 
\[
	g_{i,i+1} = \frac{(c-i)\alpha g_{i+1,i+1}}{\lambda + (c-i)\alpha + i \mu - \lambda (g_{i,i} + g_{i+1,i+1})}, \qquad i = 0,1,\dots,c-1.
\]
It should be noted that the quantities in the left hand side are given. Furthermore, comparing elements $(i,j)$ in both sides of (\ref{equationforG:eq}) and rearranging the result, we obtain 
\[
	g_{i,j} = \frac{ (c-i) \alpha + \lambda \sum_{k=i+1}^{j-1} g_{i,k} g_{k,j} }{ \lambda + (c-i)\alpha + i\mu - \lambda (g_{i,i}+g_{j,j})}, \qquad i + 1 < j \leq c.
\]

Once $G$ is given, we obtain other $G^{(n)}$ ($n = 1,2,\dots,c$) matrices using the following backward formula.  
\[
	G^{(n)} =  \left( -Q^{(n)}_0 - Q^{(n)}_1 G^{(n+1)}  \right)^{-1}   Q^{(n)}_{-1} , \qquad n = c,c-1,\dots,1.
\]

\section{Comparison of Several Approaches}\label{comparison_approach:sec}

In this section, we present a comparison between several approaches that can be used to solve our M/M/c/Setup model.
\begin{remark}
We observe that the generating function approach and the matrix analytic method are equivalent in the following sense. Indeed, the homogeneous part in the QBD formulation corresponds to $\widehat{\Pi}_i (z)$ ($i = 0,1,\dots,c$) in the generating function approach. The non-homogeneous part in the matrix analytic method corresponds to the boundary part, i.e., $\{ (i,j); j = i = 0,1,\dots, c, i \leq j  \leq c \}$ in the generating function approach. The advantage of the matrix analytic method is that it directly implies a recursive formula for computing the rate matrix. In our case, the generating function approach yields the exact closed form solution for the joint stationary distribution. 

In general, in case the stationary distribution is exactly obtainable, generating function gives detailed information of the model. On the other hand, when such an analytical solution does not exit, matrix analytic approach provides a unify approach for numerical calculation. 
\end{remark}

\begin{remark}
The matrix analytic method here shares many spirits with the recursive renewal approach. In particular, both methods are based on probabilistic arguments. For example, the quantity $p^L_{i \to d}$ in~\cite{Gandhi13,Gandhi14} is identical to $g_{i,d}$ in Section~\ref{homo_part_gmatrix}. It should be noted that matrix $R$ could be obtained from matrix $G$. From this point of view, the matrix analytic method and the recursive renewal approach are equivalent. The difference in both approaches is that while the matrix analytic method aims at a direct computation of the queue length distribution, the recursive renewal reward approach could be used to obtain any quantity of interest such as the generating function of the queue length. 
\end{remark}

\begin{remark}
Van Houdt and Leeuwaarden~\cite{Benny_Johan11} analyze a more general models, i.e., M/G/1-type and GI/M/1-type Markov chains. In~\cite{Benny_Johan11} there is only one boundary level, i.e. level 0 and thus the focus is put on the explicit expression for the $G$-matrix (or $R$-matrix). In principle, our model falls to the framework of~\cite{Benny_Johan11} by considering the non-homogeneous part as a single macro level. However, if we do so, the computational complexity in the boundary is dominant (i.e. order of $O(c^6)$) while the complexity of matrix $G$ or $R$ is only $O(c^2)$.
It should be noted that special structure of non-homogeneous part is not taken into account in~\cite{Gandhi13,Gandhi14}, thus the computational complexity is also $O(c^6)$.  
\end{remark}

\section{Some Variant Models}
In~\cite{Gandhi13}, some variants of the M/M/$c$/Setup queue are presented and analyzed. The first variant model is the M/M/$c$/Setup/Sleep where a set of $s \leq c$ servers is set to ``sleep" when idle whereas the rest $c-s$ servers are turned off when idle. The characteristic of the sleep state is that it takes a shorter setup time than the off state. The second variant is the M/M/$c$/Setup/Delayoff where a server stays idle for a while after completing a service but not yet having a job to serve. We confirm that the non-homogeneous part (the number of jobs in the system is greater than $c$) has the same structure with that 
of the M/M/$c$/Setup queue in this paper. In comparison with the original model, the boundary part M/M/$c$/Setup/Sleep has the same structure while that of M/M/$c$/Setup/Delayoff is different. The QBD formulation allows to obtain explicit rate matrix for the homogeneous part for both models using which we can recursively obtain the stationary distribution. The generating function approach in this paper can be applied to the M/M/$c$/Setup/Sleep directly while some further modifications are needed for the M/M/$c$/Setup/Delayoff model.

\section{Performance Measures and Numerical Examples}\label{numerical:sec}
\subsection{Performance measures}

Let $\pi_i$ denote the stationary probability that there are $i$ active servers, i.e., $\pi_i = \sum_{j = i}^\infty \pi_{i,j}$. 
Let $\mathbb{E} [A] $ and $\mathbb{E} [S] $ denote the mean number of active servers and that in setup mode, respectively. We have 
\[
	\mathbb{E} [A]  = \sum_{i = 1}^c i \pi_i , \qquad \mathbb{E} [S] = \sum_{i = 0}^c \sum_{j = i}^\infty  \min(j-i, c-i) \pi_{i,j}.
\]
Let $\mathbb{E} [S_r]$ denote the switching rate from OFF to ON in the steady state (mean number of switches from OFF to ON per unit time). We then have 
\begin{align*}
	\mathbb{E} [S_r] &    = \sum_{i=0}^{c-1} \sum_{j = i+1}^\infty \min(c-i,j-i) \alpha \pi_{i,j}   = \sum_{i = 1}^c i \mu \pi_{i,i}, 
\end{align*}
where the second equality is due to the fact that the switching rate from OFF to on is equal to that from ON to OFF in the steady state. 
Furthermore, let $\mathbb{E} [L]$ denote the mean number of jobs in the systems, i.e., 
\[
	\mathbb{E} [L] = \sum_{j=0}^\infty \vc{\pi}_j \vc{e} j,
\]
where $\vc{\pi}_j \vc{e}$ is the probability that there are $j$ customers in the system.

We define a cost function for the model. 
%
\[
	Cost_{on-off} = C_a \mathbb{E} [A] + C_s \mathbb{E} [S].
\]
where $C_a$ and $C_s$ are the cost per time unit for an active server and a server in setup mode, respectively.

For comparison, we also define the cost of the corresponding ON-IDLE model, i.e., M/M/$c$ without setup times. It is easy to see that the power consumption for this model is given as follows. 
\[
	Cost_{on-idle}  = c \rho C_a + c (1 - \rho) C_i,
\]
where $C_i$ is the power consumption of an idle server.

If each time of turning ON and turning OFF a server needs a cost of $C_{sw}$, we could also consider the following cost function~\cite{Maccio13}.
\[
	TotalCost_{on-off} = C_a \mathbb{E} [A] + C_s \mathbb{E} [S] + C_{sw} \mathbb{E} [S_r].
\]

\subsection{Numerical examples}

In this section, we show some numerical examples. It should be noted that some of them are also presented in~\cite{Gandhi10,Gandhi13,Gandhi14}. The numerical results are presented to show the feasibility of our computational procedure. Furthermore, we complement numerical results in~\cite{Gandhi10,Gandhi13,Gandhi14} by taking the switching rate between ON and OFF into account. 

In all the numerical examples, we fix $\mu = 1$, $C_a = C_s = 1$ and $C_i = 0.6 C_a$. The evidence for $C_i = 0.6 C_a$ is that an idle server still consumes about 60\% of its peak processing a job~\cite{Barroso07}. We will investigate the cost function with respect to the setup cost $C_s$ in Section~\ref{cost_vs_setup_cost:sec}. 

All the numerical results in this section are obtained using the matrix analytic method presented in Section~\ref{matrix_ana:sec}. The same numerical results can be also obtained using the procedure presented in Section~\ref{components:sec}.
\subsubsection{Effect of the setup rate}\label{setup_rate:subsec}

Section~\ref{setup_rate:subsec} investigates the effect of the setup rate on the power consumption ($Cost_{on-off}$, $Cost_{on-idle}$) and the mean number of jobs in the system. Figures~\ref{Power_consump_vs_nu_c20:fig} and \ref{Power_consump_vs_nu_c30:fig} represent the power consumption against the setup rate for the case $c=20$ and 30, respectively. We observe that the power consumption decreases as the setup rate increases. For comparison, we also plot the power consumption for the corresponding M/M/$c$ model without setup times. We find that there exists some $\alpha_{\rho, c}$ such that the ON-OFF policy outperforms the ON-IDLE policy for $\alpha > \alpha_{\rho, c}$ while the latter is more power-saving for the case $\alpha < \alpha_{\rho, c}$. Furthermore, $\alpha_{\rho, c}$ increases as $\rho$ increases.

Figures~\ref{Power_consump_vs_nu_switching_cost_c20:fig} and~\ref{Power_consump_vs_nu_switching_cost_c30:fig} investigate the total energy consumption taking into account the switching cost, i.e., $TotalCost_{on-off}$ ($C_{sw} = 1$) against the setup rate $\alpha$ for $\rho = 0.3, 0.5$ and 0.7. We observe that the total power consumption does not always monotonically decreases as the setup rate increases as in~Figures~\ref{Power_consump_vs_nu_c20:fig} and \ref{Power_consump_vs_nu_c30:fig}. This is because when the setup rate $\alpha$ is large the number of switches per time unit increases leading to the increase in the cost function. We observe in the curves of $\rho = 0.5$ that there exist two points $\alpha_{min} $ and $\alpha_{max} $ such that the ON-IDLE policy outperforms the ON-OFF policy for $\alpha < \alpha_{min}$ and $\alpha > \alpha_{max}$. An interesting observation is that three curves for $\rho = 0.3, 0.5$ and 0.7 are the same when the setup rate is extremely low. The reason is that all the servers are in setup mode for almost the time when the setup time is extremely long.

\begin{figure}[htbp]
\begin{tabular}{cc}
\begin{minipage}{0.5\hsize}
\begin{center}
\includegraphics[scale=0.55]{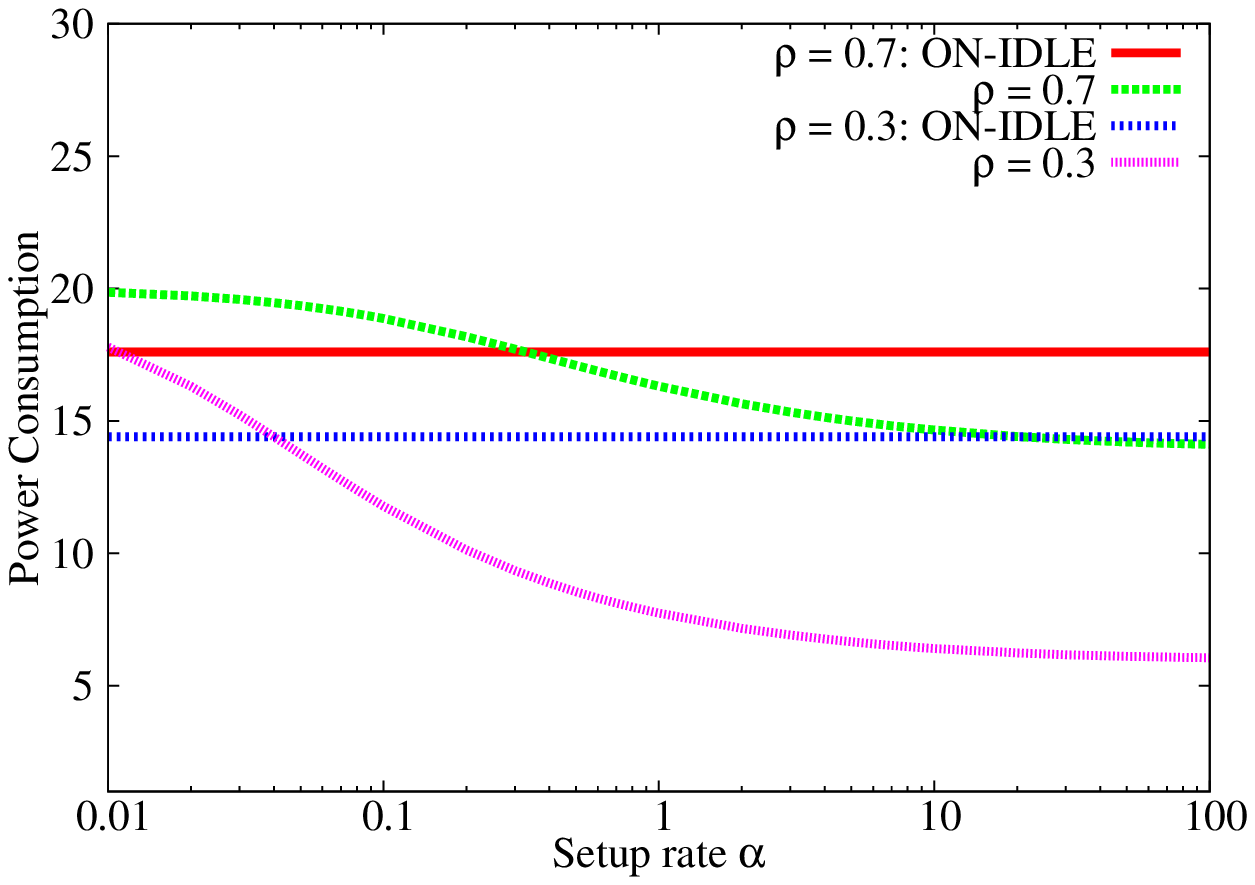}
\caption{Power consumption vs. $\alpha$ ($c=20$)}
\label{Power_consump_vs_nu_c20:fig}
\end{center}
\end{minipage}
\begin{minipage}{0.5\hsize}
\begin{center}
\includegraphics[scale=0.55]{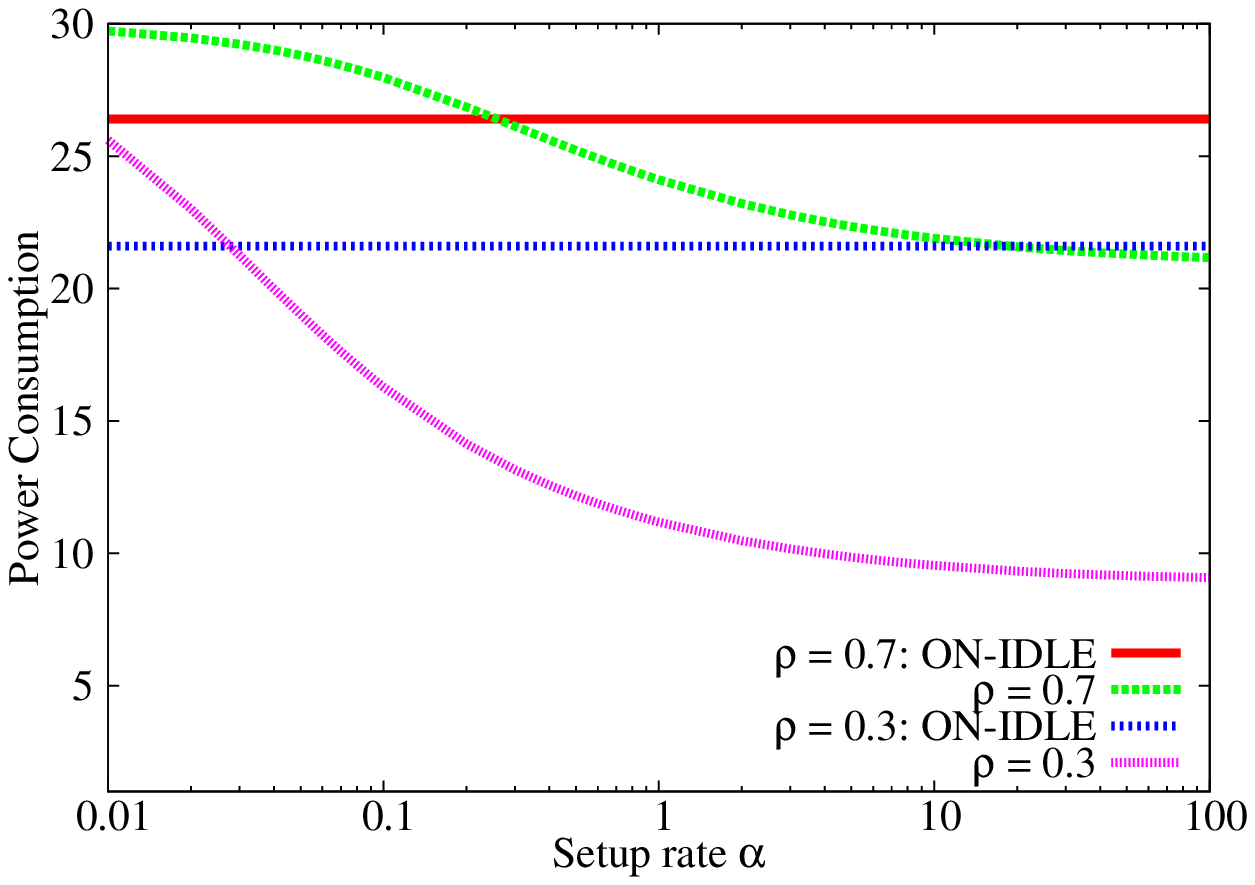}
\caption{Power consumption vs. $\alpha$ ($c=30$)}
\label{Power_consump_vs_nu_c30:fig}
\end{center}
\end{minipage}
\end{tabular}
\end{figure}

\begin{figure}[htbp]
\begin{tabular}{cc}
\begin{minipage}{0.5\hsize}
\begin{center}
\includegraphics[scale=0.55]{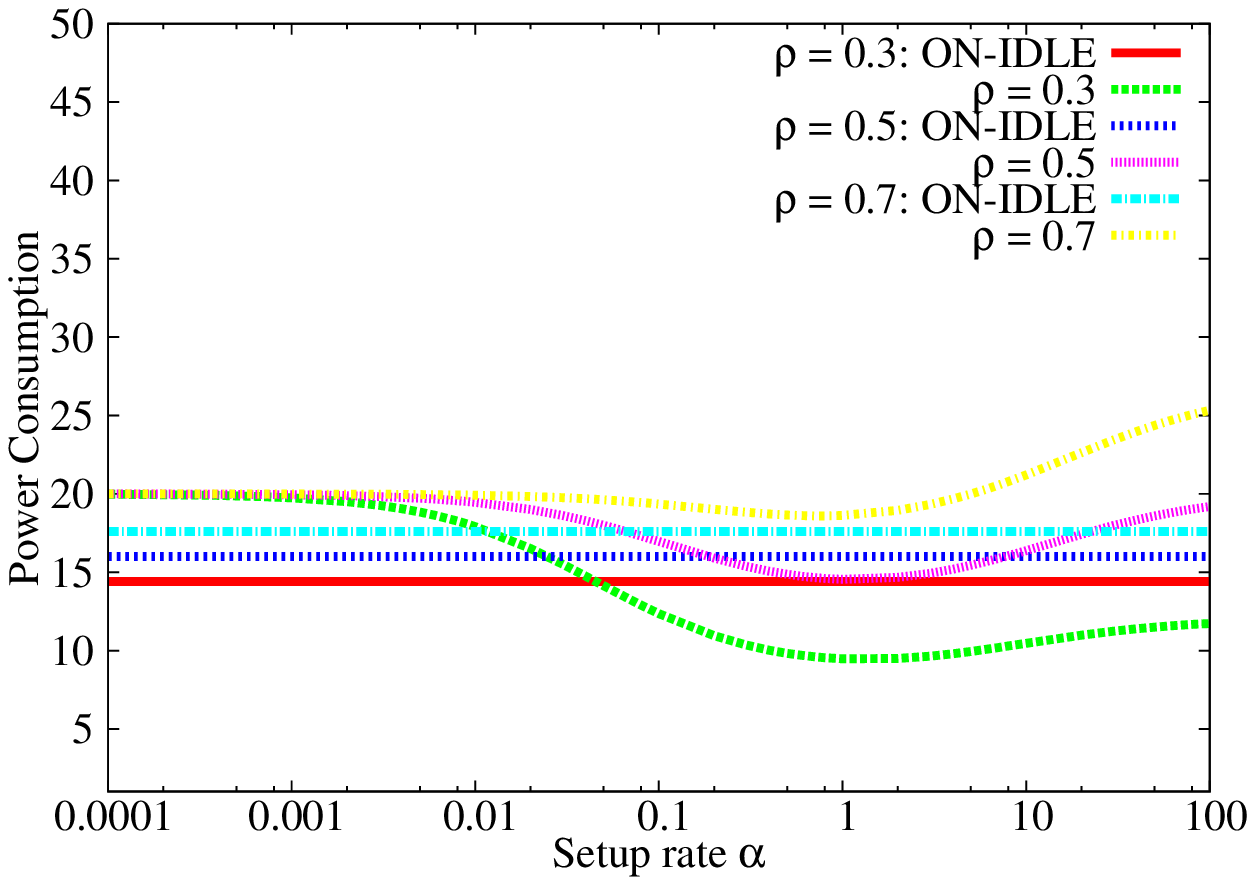}
\caption{Total Pow. consump. vs. $\alpha$ ($c=20$)}
\label{Power_consump_vs_nu_switching_cost_c20:fig}
\end{center}
\end{minipage}
\begin{minipage}{0.5\hsize}
\begin{center}
\includegraphics[scale=0.55]{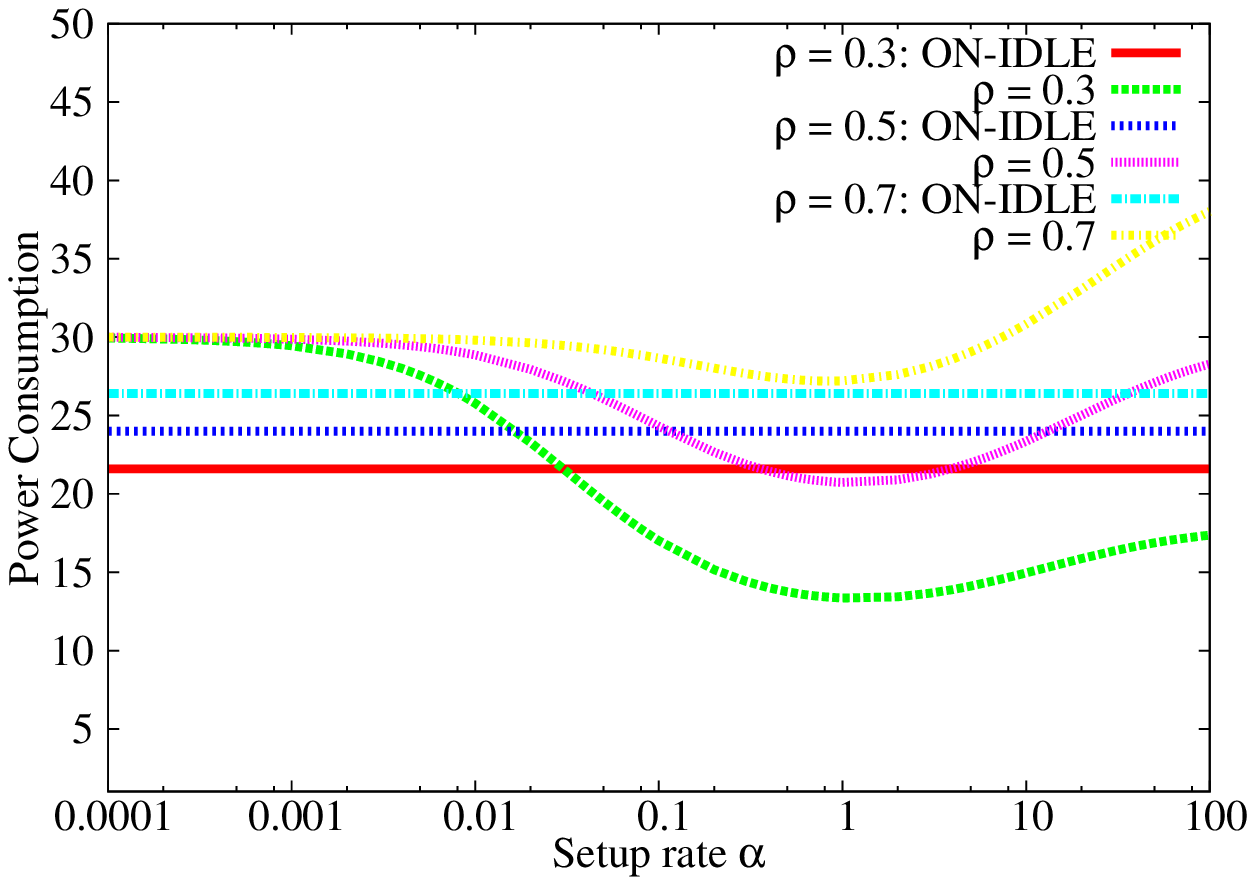}
\caption{Total Pow. consump. vs. $\alpha$ ($c=30$)}
\label{Power_consump_vs_nu_switching_cost_c30:fig}
\end{center}
\end{minipage}
\end{tabular}
\end{figure}

Figures~\ref{System_num_vs_nu_c10:fig} and \ref{System_num_vs_nu_c30:fig} represent the mean number of jobs in the system ($\mathbb{E} [L]$) against the setup rate $\alpha$. We observe that $\mathbb{E} [L]$ decreases as the setup rate increases. We also observe that $\mathbb{E} [L]$ converges to that of  the ON-IDLE model as $\alpha \to \infty$ which agrees with intuition.

\begin{figure}[htbp]
\begin{tabular}{cc}
\begin{minipage}{0.5\hsize}
\begin{center}
\includegraphics[scale=0.55]{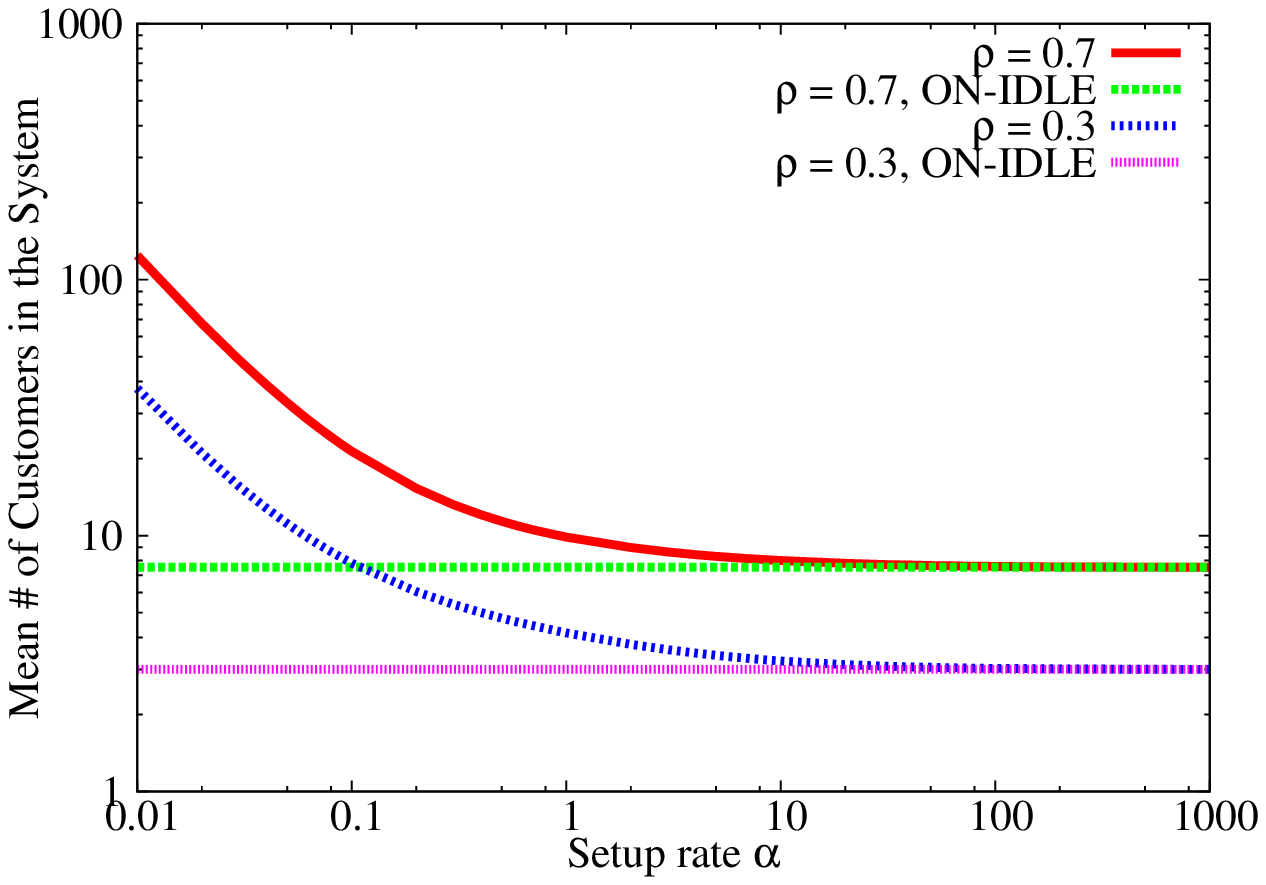}
\caption{$\mathbb{E} [L]$ vs. $\alpha$ ($c=10$)}
\label{System_num_vs_nu_c10:fig}
\end{center}
\end{minipage}
\begin{minipage}{0.5\hsize}
\begin{center}
\includegraphics[scale=0.55]{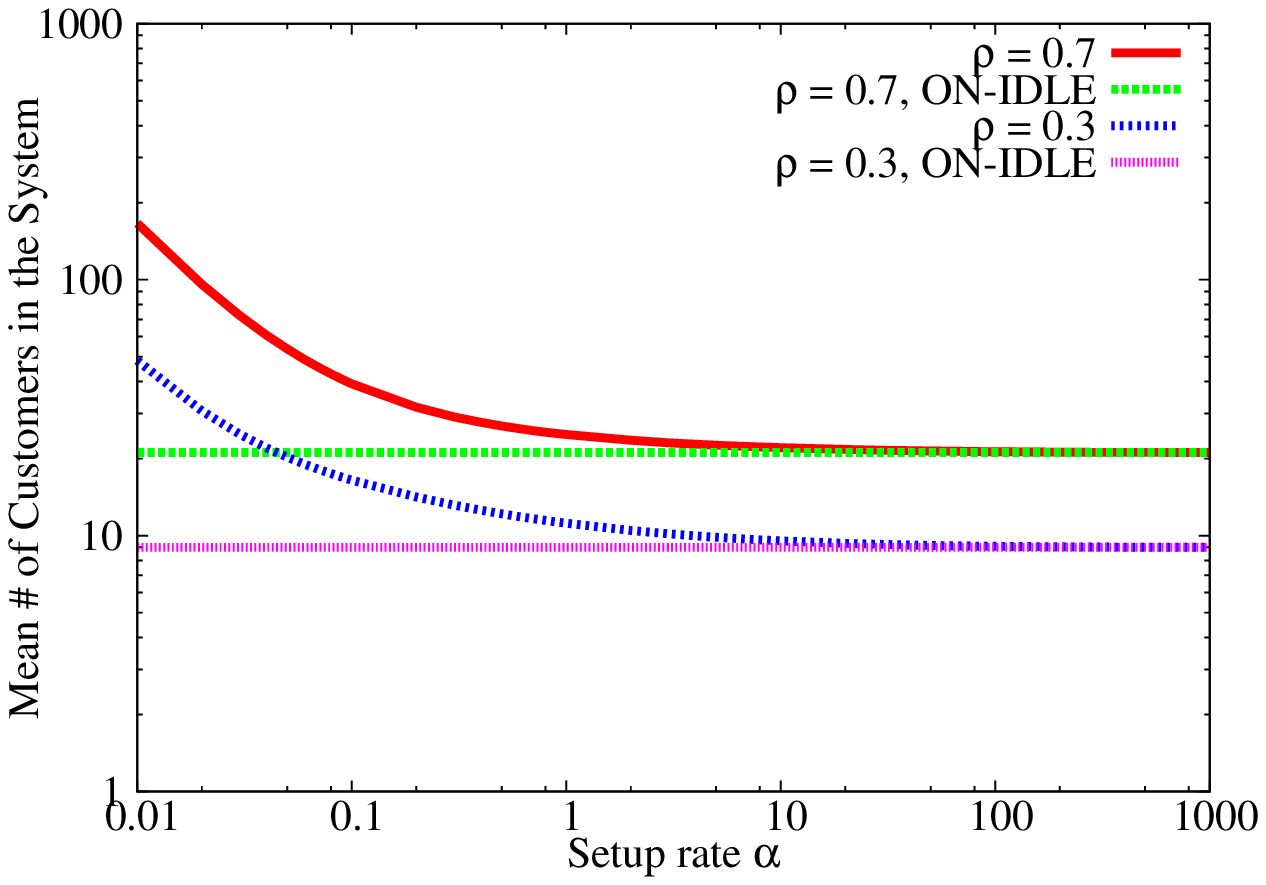}
\caption{$\mathbb{E} [L]$ vs. $\alpha$ ($c=30$)}
\label{System_num_vs_nu_c30:fig}
\end{center}
\end{minipage}
\end{tabular}
\end{figure}

\subsubsection{Effect of the number of servers}
In this subsection, we investigate the effect of the number of servers on the power consumption ($Cost_{on-off}$) while keeping the traffic intensity for each server, i.e., $\rho = \lambda/(c\mu)$ constant. Figures~\ref{Power_consump_vs_serversnum_rho05:fig} and~\ref{Power_consump_vs_serversnum_rho07:fig} represent the case $\rho = 0.5$ and $\rho = 0.7$, respectively. We observe in both figures that the ON-OFF policy is always more power-saving than the ON-IDLE policy for $\alpha = 1$ while the latter always outperforms the former for the case $\alpha = 0.01$. For the case $\alpha = 0.1$, we observe in Figure~\ref{Power_consump_vs_serversnum_rho05:fig} that there exists some $c_{\alpha = 0.1}$ such that the ON-OFF policy outperforms the ON-IDLE one for $c > c_{\alpha = 0.1}$ while the latter is more power-saving than the former for $c < c_{\alpha = 0.1}$. Thus, for $\alpha=0.1$ and $\rho = 0.5$, the ON-OFF policy is more effective than the ON-IDLE system if the scale of the system is large enough.

\begin{figure}[htbp]
\begin{tabular}{cc}
\begin{minipage}{0.5\hsize}
\begin{center}
\includegraphics[scale=0.55]{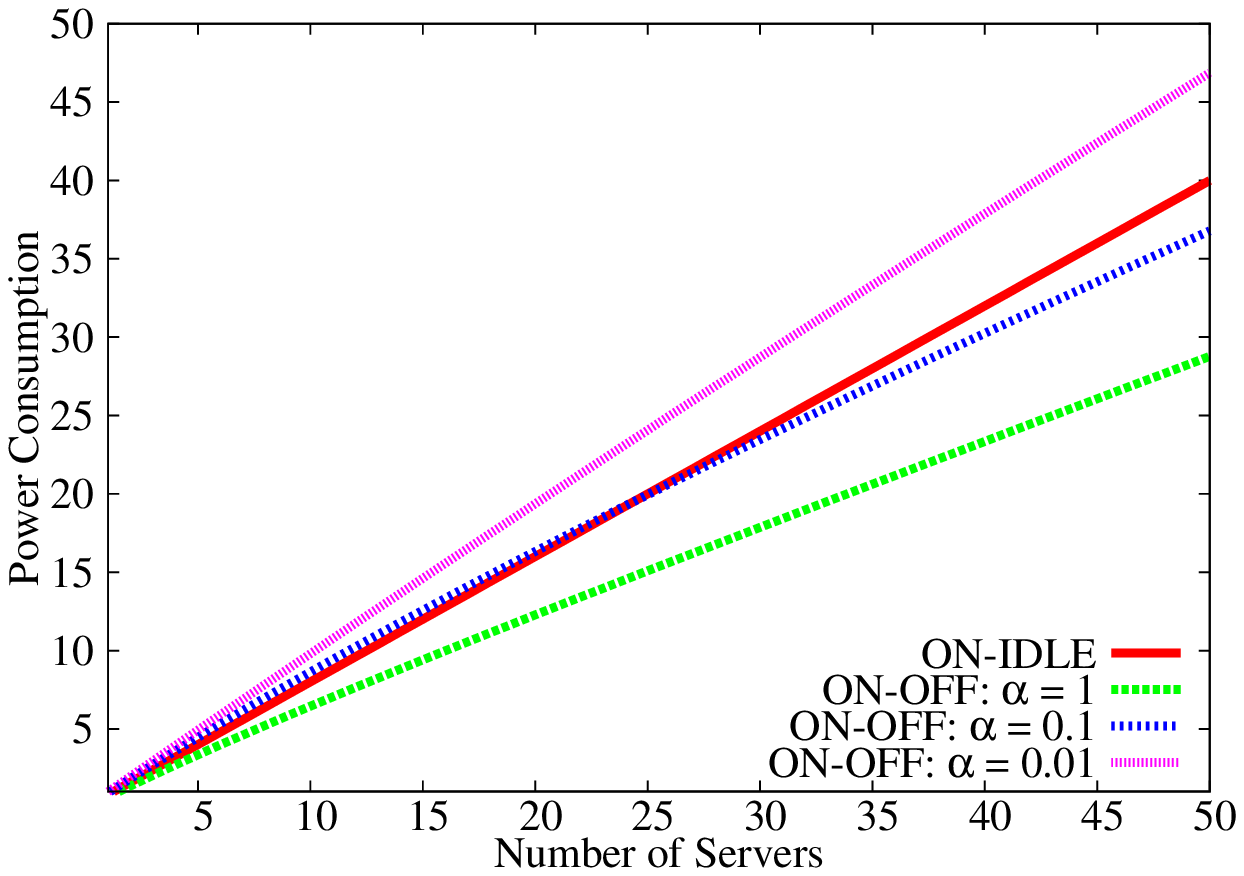}
\caption{Power consumption vs. $c$ ($\rho = 0.5$).}
\label{Power_consump_vs_serversnum_rho05:fig}
\end{center}
\end{minipage}
\begin{minipage}{0.5\hsize}
\begin{center}
\includegraphics[scale=0.55]{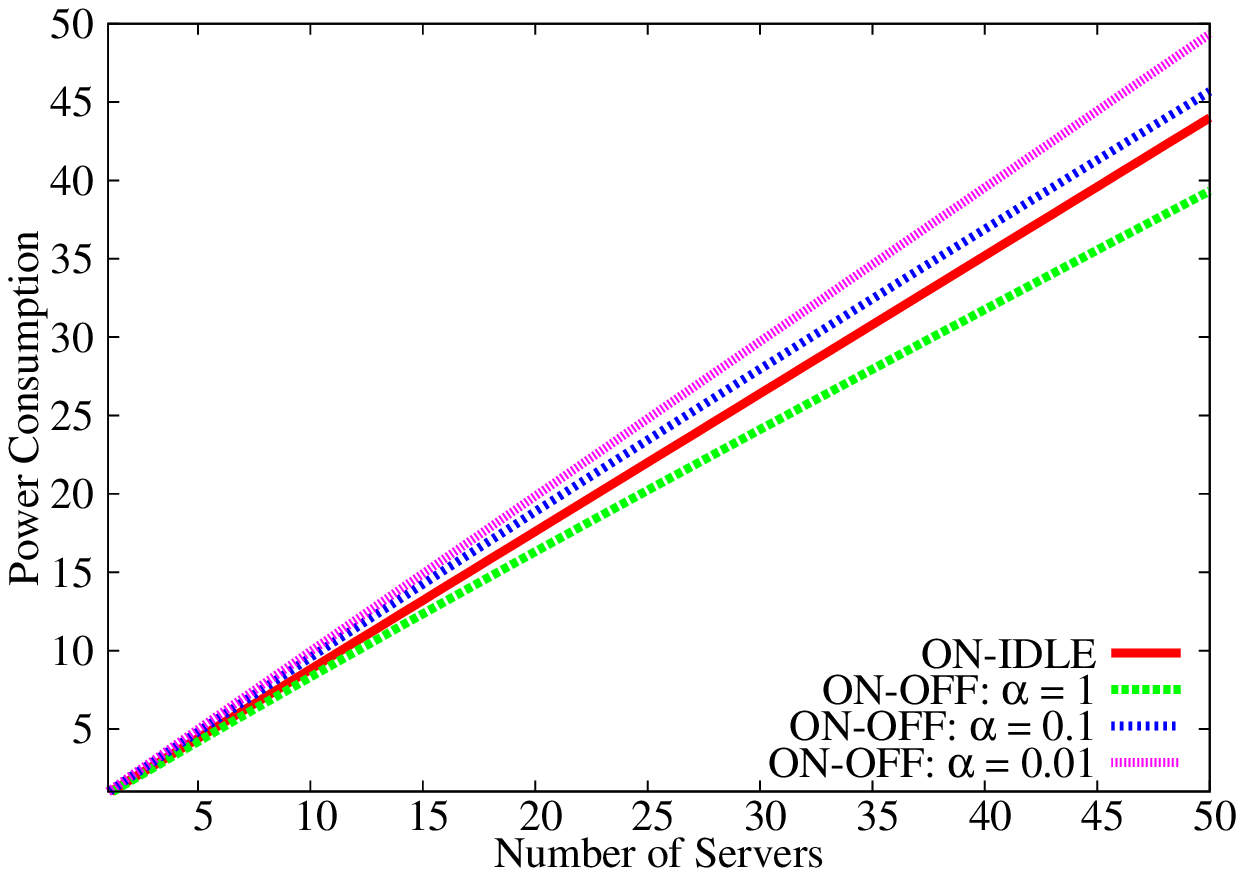}
\caption{Power consumption vs. $c$ ($\rho = 0.7$).}
\label{Power_consump_vs_serversnum_rho07:fig}
\end{center}
\end{minipage}
\end{tabular}
\end{figure}

\subsubsection{Effect of traffic intensity}
In this section, we show the effect of the traffic intensity on the power consumption ($Cost_{on-off}, Cost_{on-idle},$) for the cases $c=20$ and $c=50$ in Figure~\ref{Power_consump_vs_traffic_intensity_c20:fig} and Figure~\ref{Power_consump_vs_traffic_intensity_c50:fig}, respectively. In each figure, we plot three curves with $\alpha = 1, 0.1$ and 0.01. For comparison, we also plot the power consumption for the corresponding model without setup times. We observe in both figures that the ON-OFF policy with $\alpha = 1$ always outperforms that of ON-IDLE policy. However,  for the cases $\alpha = 0.1$ and 0.01, we observe that there exists some $\rho_\alpha$ for which the ON-OFF policy outperforms the ON-IDLE one for $\rho < \rho_\alpha$ while the latter is more power-saving than the former for the case $\rho > \rho_\alpha$.

\begin{figure}[htbp]
\begin{tabular}{cc}
\begin{minipage}{0.5\hsize}
\begin{center}
\includegraphics[scale=0.55]{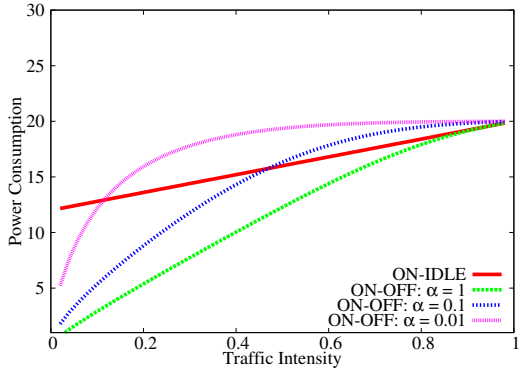}
\caption{Power consumption vs. $\rho$ ($c = 20$).}
\label{Power_consump_vs_traffic_intensity_c20:fig}
\end{center}
\end{minipage}
\begin{minipage}{0.5\hsize}
\begin{center}
\includegraphics[scale=0.55]{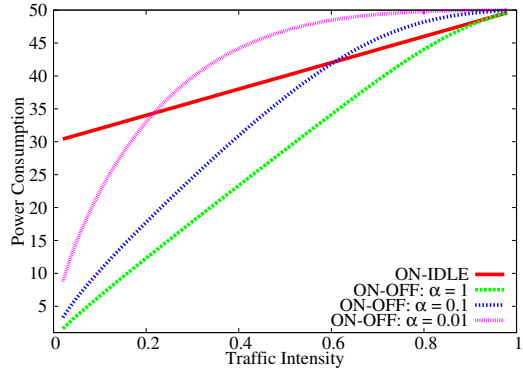}
\caption{Power consumption vs. $\rho$ ($c = 50$).}
\label{Power_consump_vs_traffic_intensity_c50:fig}
\end{center}
\end{minipage}
\end{tabular}
\end{figure}

\subsubsection{Effect of the setup cost}\label{cost_vs_setup_cost:sec}
Figure~\ref{pefor_vs_SetupCostofOneSwitch:fig} show the sensitivity of the cost of a setting up server on the power consumption $Cost_{on-off}$ where $C_a = 1$. Letting $r = C_s/C_a$, we observe that there exists some $r_{\rho}$ such that the ON-IDLE policy outperforms the ON-OFF policy for $r > r_\rho$ while former outperforms the latter for the case $r < r_\rho$. We also observe that $r_\rho$ decreases with the increase of $\rho$. This agrees with intuition. 

Figure~\ref{Total_Power_consump_vs_traffic_intensity_c50:fig} represents the total power consumption ($TotalCost_{on-off}$ with $C_{sw} = 1$) against the traffic intensity. We observe in the curves of $\alpha = 0.01, 0.1$ and 1 that the total power consumption monotonically increases as the traffic intensity increases. Interestingly, we observe that for the case $\alpha = 10$ and 100, the total power consumption increases as $\rho$ increases (for a relatively small $\rho$) and then decreases as $\rho$ increases (for a relatively large $\rho$). At the first glance, it may not be intuitive that the total power consumption decreases with the increase in $\rho$. However, this is due to the relation of $\mathbb{E} [S_r]$ and $\rho$ which will be investigated in detail in Figure~\ref{Number_of_Switches_c4050:fig}.

\begin{figure}[htbp]
\begin{tabular}{cc}
\begin{minipage}{0.5\hsize}
\begin{center}
\includegraphics[scale=0.55]{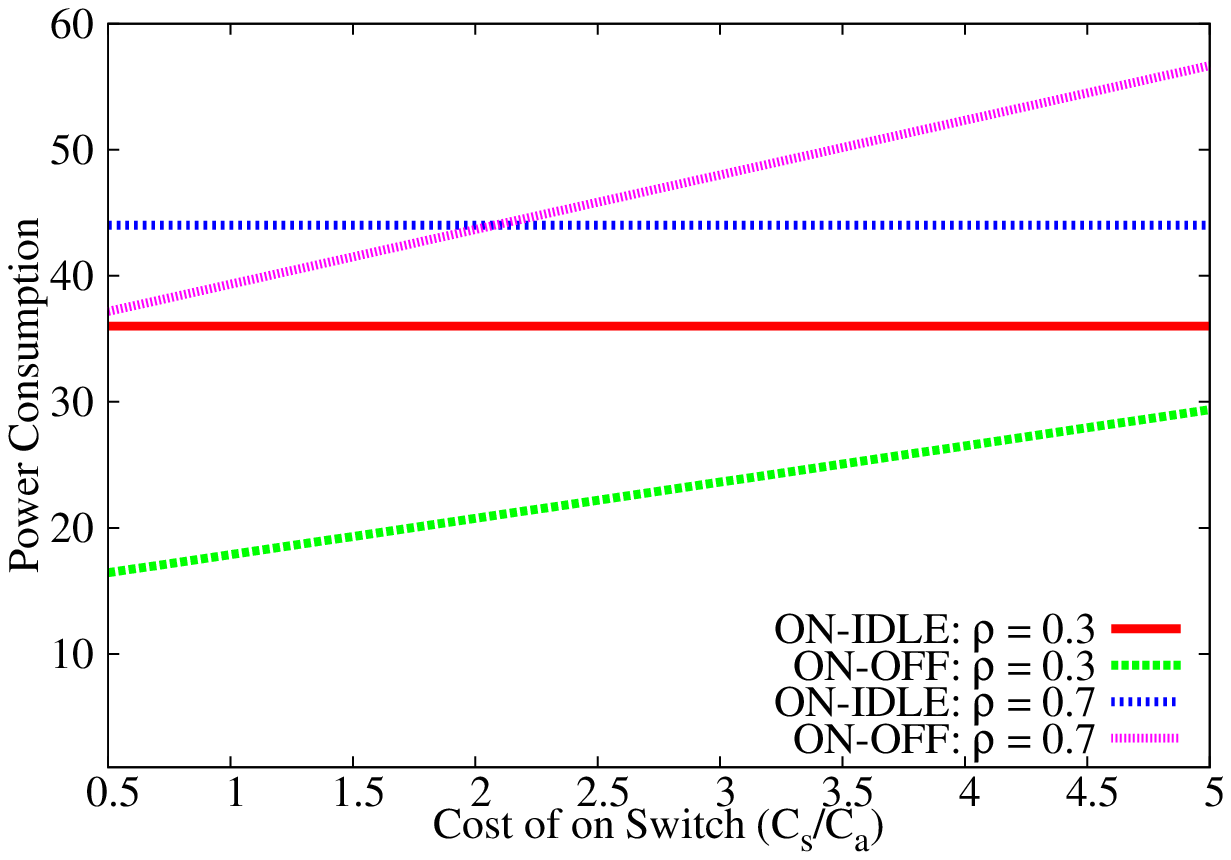}
\caption{Power consumption vs. $C_s/C_a$.}
\label{pefor_vs_SetupCostofOneSwitch:fig}
\end{center}
\end{minipage}
\begin{minipage}{0.5\hsize}
\begin{center}
\includegraphics[scale=0.55]{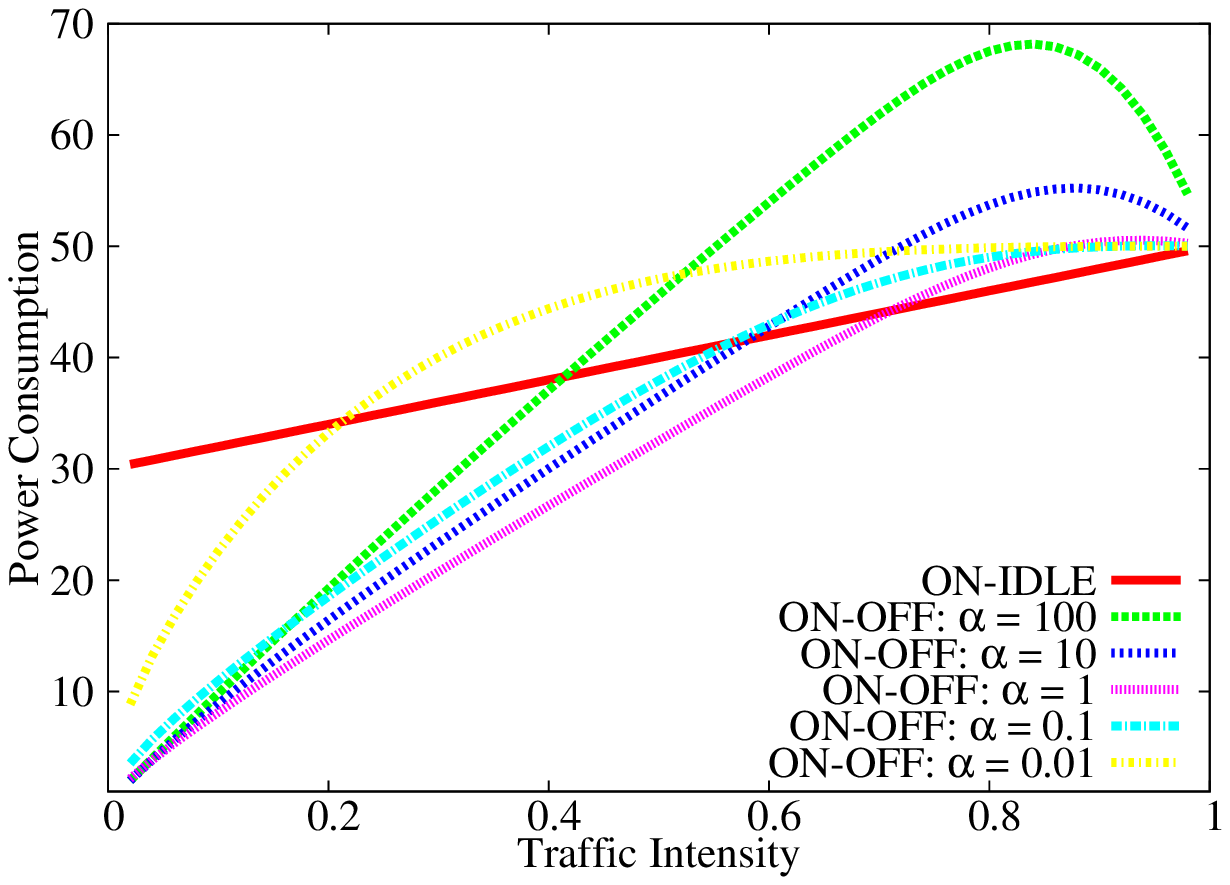}
\caption{Total Pow. consump. vs. $\rho$ ($c = 50$).}
\label{Total_Power_consump_vs_traffic_intensity_c50:fig}
\end{center}
\end{minipage}
\end{tabular}
\end{figure} 
\begin{figure}[htbp]
\end{figure}

\subsubsection{Mean number of switches}
In this section, we investigate the property of the switching rate $\mathbb{E} [S_r]$, i.e., the mean number of switches per a time unit. In particular, Figure~\ref{Number_of_Switches_c4050:fig} shows the switching rate against the traffic intensity. We observe that the switching rate increases with the traffic intensity under a light traffic regime while it decreases with $\rho$ in relatively heavy traffic regime. The reason is as follows. Almost all the servers are OFF in light traffic regime while a large percent of servers are ON in heavy traffic. Thus, in light traffic regime, increasing the traffic intensity implies the increase in the number of switches from OFF to ON. However, in heavy traffic regime almost all the servers are already ON. As a result, increasing the traffic intensity does not lead to further increase in the switching rate. This suggests that from the switching rate point of view, the ON-OFF policy is preferable in a relatively light traffic regime or a relatively heavy traffic one.

Figure~\ref{Number_of_Switches_vs_serversnum_rho07:fig} shows the switching rate against the number of servers. We observe that the switching rate increases with the number of servers. Moreover, the curves for the case $\alpha = 0.1$ and $\alpha = 0.01$ are almost linear while that for the case $\alpha = 1$ is not linear.

\begin{figure}[htbp]
\begin{tabular}{cc}
\begin{minipage}{0.5\hsize}
\begin{center}
\includegraphics[scale=0.55]{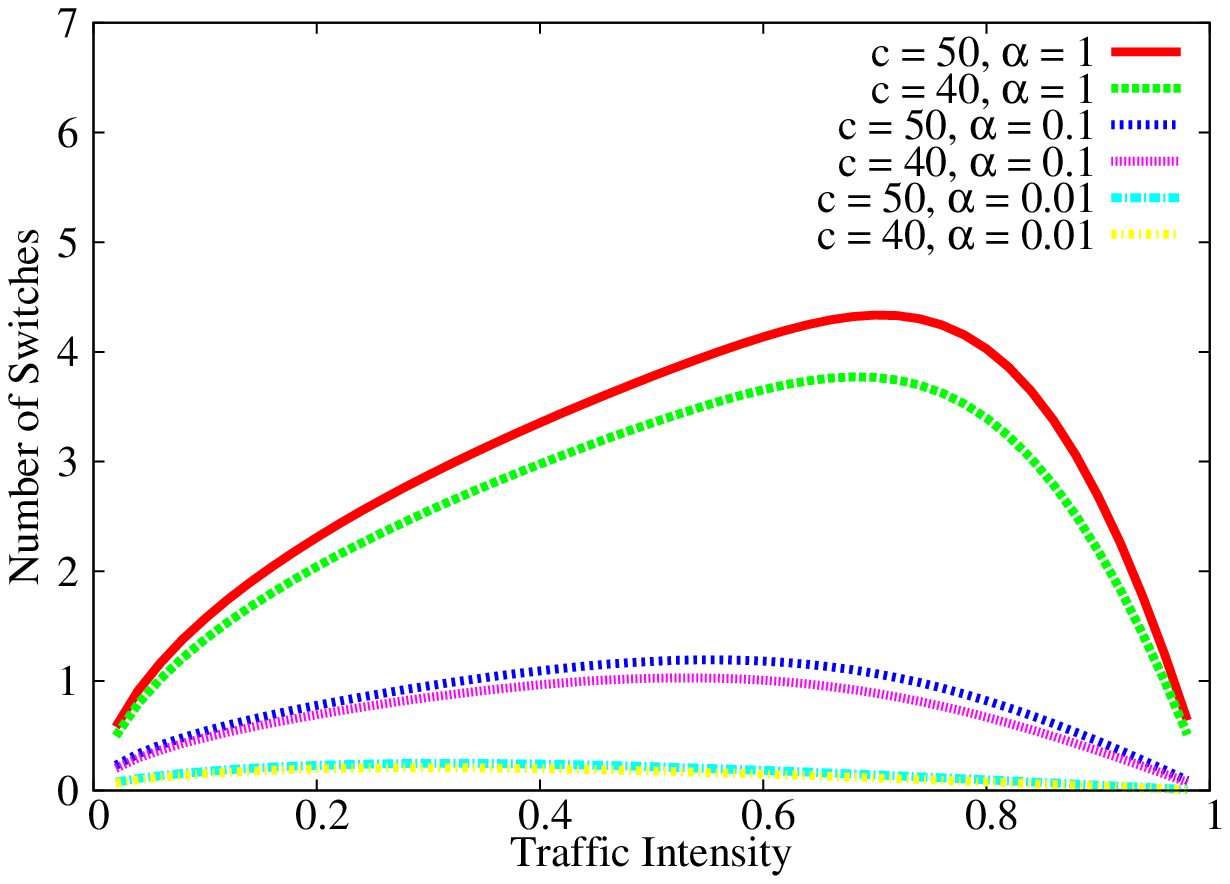}
\caption{Switching rate vs. $\rho$ ($c = 40, 50$).}
\label{Number_of_Switches_c4050:fig}
\end{center}
\end{minipage}
\begin{minipage}{0.5\hsize}
\begin{center}
\includegraphics[scale=0.55]{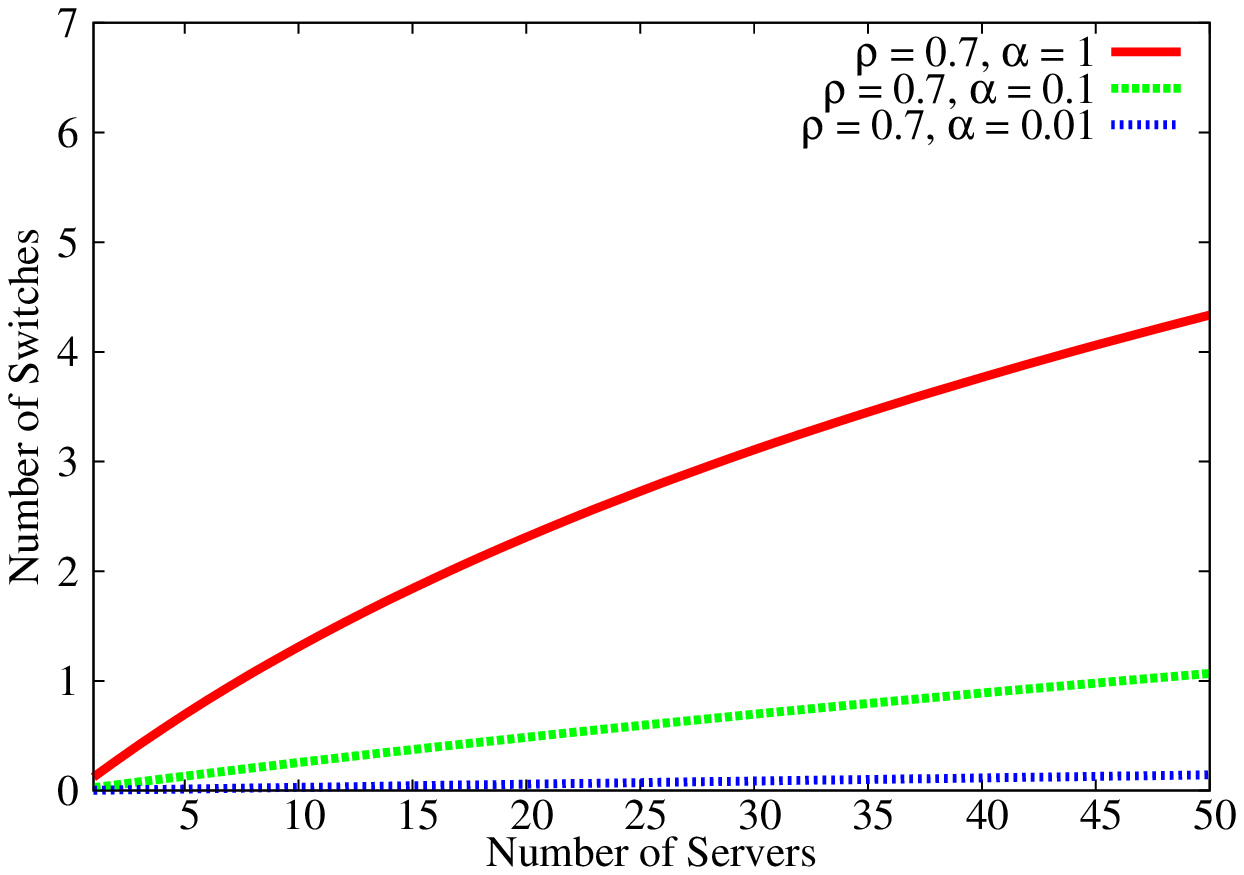}
\caption{Switching rate vs. the $c$ ($\rho = 0.7$).}
\label{Number_of_Switches_vs_serversnum_rho07:fig}
\end{center}
\end{minipage}
\end{tabular}
\end{figure}

\section{Conclusion and future works}\label{conclusion:sec}
In this paper, we have presented a detailed analysis for the M/M/$c$/Setup model with ON-OFF policy for data centers. Using a generating function approach, we have derived explicit solutions for the generating functions from which we have obtained recursive formulae for the factorial moments. The generating function approach yields a conditional decomposition for the queue length. We also have observed that the model belongs to a special QBD class where the rate matrix of the homogeneous part is explicitly obtained. The boundary part also possesses some special structure allowing us to obtain the joint stationary distribution with the complexity of $O(c^2)$ by generating function approach and $O(c^3)$ by the matrix analytic method. Our numerical results have provided some insights into the performance of the system. We have found the range of the parameters under which the ON-OFF policy outperforms the ON-IDLE policy. We have pointed out the equivalence between the two methodologies. 

In real world data center, in order to reduce the waiting time, a fixed number of servers may be kept ON all the time. The extension of the current model to this case may be worth to investigate. Other extensions include a threshold policy which turns ON and OFF the servers according to the load of the system.

\section*{Acknowledgements}
The author would like to thank two anonymous referees and the associate editor whose comments helped to improve the presentation of the paper.
The author would like to thank Professor Herwig Bruneel of Ghent University and Professor Onno Boxma of Eindhoven University of Technology for useful remarks on the conditional decomposition. This research was supported in part by Japan Society for the Promotion of Science, JSPS Grant-in-Aid for Young Scientists (B), Grant Number 2673001.

\end{document}